\documentclass[a4paper]{article}

\usepackage{a4wide}
\usepackage[T1]{fontenc}
\usepackage[utf8]{inputenc}
\usepackage{authblk}

\usepackage[table]{xcolor}
\definecolor{lilla}{HTML}{750787}
\usepackage{amssymb,amsmath,amsthm}
\usepackage{hyperref}
\usepackage[ruled,vlined,linesnumbered]{algorithm2e}
\newcommand{\myIf}[2]{\textbf{if}~#1~\textbf{then}~#2\;}
\usepackage[capitalise]{cleveref}
\usepackage{graphicx}
\usepackage{caption}
\usepackage{subcaption}
\usepackage{booktabs}
\usepackage{tabularx}
\usepackage{multicol}
\usepackage{gensymb}
\usepackage{enumerate}   
\usepackage{csquotes}
\usepackage[hyperpageref]{backref}
\usepackage{paralist}

\renewcommand*{\backref}[1]{}
\renewcommand*{\backrefalt}[4]{\ifcase #1\or [p.~#2.]\else [pp.~#2.]\fi }

\hypersetup{
	colorlinks,citecolor=green!50!black,linkcolor=red!50!black,urlcolor=black
}

\usepackage{tikz}

\usepackage{mathtools}
\usepackage{xargs}
\newcommandx{\set}[2][1=1]{\ensuremath{\{#1,\ldots,#2\}}}
\newcommandx{\tlog}[3][1=,3=]{\log_{#1}^{#3}(#2)}
\newcommandx{\ith}[2][1=th]{#2\nobreakdash-#1}

\newtheorem{theorem}{Theorem}
\newtheorem{lemma}[theorem]{Lemma}

\newtheorem{proposition}[theorem]{Proposition}
\newtheorem{observation}[theorem]{Observation}
\newtheorem{corollary}[theorem]{Corollary}
\theoremstyle{definition}
\newtheorem{problem}{Problem}
\newtheorem{rrule}{Reduction Rule}

\newtheorem{construction}{Construction}
\newcommand{\cqed}{\hfill$\diamond$}

\crefname{observation}{Observation}{Observations}
\crefname{rrule}{Reduction Rule}{Reduction Rules}
\crefname{construction}{Construction}{Constructions}
\crefname{theorem}{Theorem}{Theorems}
\Crefname{theorem}{Thm.}{Thms.}
\crefname{corollary}{Corollary}{Corollaries}
\crefname{lemma}{Lemma}{Lemmata}
\Crefname{corollary}{Cor.}{Cors.}
\crefname{proposition}{Proposition}{Propositions}
\Crefname{proposition}{Prop.}{Props.}
\crefname{algorithm}{Algorithm}{Algorithms}

\newcommand{\calT}{\mathcal{T}}
\newcommand{\calX}{\mathcal{X}}
\newcommand{\yes}{\texttt{yes}}
\newcommand{\no}{\texttt{no}}
\newcommand{\RD}{$(\Rightarrow)\:$}
\newcommand{\LD}{$(\Leftarrow)\:$}

\newcommand{\sump}{{\Sigma}}
\newcommand{\maxp}{{\infty}}
\newcommand{\undp}{{U+\tau}}

\DeclareMathOperator{\tw}{tw}

\DeclareMathOperator{\fes}{fes}
\DeclareMathOperator{\bw}{bw}
\DeclareMathOperator{\dcc}{dcc}
\DeclareMathOperator{\ncc}{ncc}
\DeclareMathOperator{\is}{is}
\DeclareMathOperator{\dom}{dom}
\DeclareMathOperator{\dcl}{dcl}
\DeclareMathOperator{\dco}{dco}
\DeclareMathOperator{\cdi}{cdi}
\DeclareMathOperator{\vc}{vc}
\DeclareMathOperator{\fvs}{fvs}
\DeclareMathOperator{\clw}{clw}

\DeclareMathOperator{\dgn}{dgn}
\DeclareMathOperator{\dbi}{dbi}
\DeclareMathOperator{\Part}{Part}

\crefname{problem}{Problem}{Problems}
\Crefname{problem}{Prob.}{Probs.}

\newcommandx{\decprob}[6][3=Input,5=Question]{\begin{samepage}
  \begingroup
\begin{problem}\label{prob:#2}{\setlength{\fboxsep}{1pt}\colorbox{gray!12!white}{\textsc{#1}}}
  \nopagebreak[4]\end{problem}\nopagebreak[4]\vspace{-0.6em}
  \par\noindent\hangindent=\parindent\textbf{#3}:  #4\nopagebreak[4]
  \par\noindent\hangindent=\parindent\textbf{#5}:  #6
  \par\medskip
  \endgroup
  \end{samepage}
}

\usepackage{dsfont}
\newcommand{\N}{\mathbb{N}}
\newcommand{\Nzero}{\N_0}

\newcommand{\bigO}{\mathcal{O}}
\renewcommand{\O}{\bigO}
\newcommand{\I}{\mathcal{I}}

\newcommand{\prob}[1]{{\normalfont\textsc{#1}}}

\newcommand{\mscTsc}{\prob{Multistage 2-Coloring}}
\newcommand{\mscAcr}{\prob{MS2C}}
\newcommand{\msceTsc}{\prob{Multistage 2-Coloring Extension}}
\newcommand{\msceAcr}{\prob{MS2CE}}
\newcommand{\mscgbTsc}{\prob{Multistage 2-Coloring on a Global Budget}}
\newcommand{\mscgbAcr}{\prob{MS2CGB}}

\newcommand{\TG}{\mathcal{G}}

\newcommand{\cocl}[1]{\ensuremath{\operatorname{#1}}}
\newcommand{\W}[1]{\cocl{W[#1]}}

\newcommand{\NP}{\cocl{NP}}
\newcommand{\FPT}{\cocl{FPT}}
\newcommand{\fpt}{fixed-parameter tractable}
\newcommand{\coNP}{\cocl{coNP}}
\newcommand{\XP}{\cocl{XP}}
\newcommand{\poly}{\cocl{poly}}

\newcommand{\NPincoNPslashpoly}{\ensuremath{\NP\subseteq \coNP/\poly}}

\newcommand{\ETHbreaks}{the ETH fails}

\newcommand{\calC}{\mathcal{C}}

\newcommand{\calR}{\mathcal{R}}

\newcommand{\ceq}{\ensuremath{\coloneqq}}

\DeclarePairedDelimiter{\abs}{\lvert}{\rvert}

\newcommand{\wilog}{without loss of generality}
\newcommand{\Wilog}{Without loss of generality}

\usepackage{etoolbox}

\clearpage{}\RequirePackage{xargs}     

\usetikzlibrary{decorations.pathreplacing}
\usetikzlibrary{math}
\usetikzlibrary{calc}
\usetikzlibrary{shapes}

\newcommand{\tikzpreamble}{\def\teps{0.075}
  \def\nsc{0.33}
  \def\tanS{1.7320508}
  \def\bsc{0.075}
\tikzstyle{xnode}=[circle,scale=\nsc,draw,fill=white];
  \tikzstyle{xnodeA}=[circle,scale=\nsc,fill=white,draw=red];
  \tikzstyle{xnodeB}=[circle,scale=\nsc,fill=lightgray,draw=green];
  \tikzstyle{xnodeC}=[circle,scale=\nsc,fill=black,draw=blue];
  \tikzstyle{xnodex}=[circle,fill,scale=\nsc,draw];
  \tikzstyle{xnodey}=[diamond,fill,scale=\nsc,draw];
  \tikzstyle{xstarL}=[star,star points=6,draw,scale=0.4];
  \tikzstyle{xstar}=[star,star points=8,draw,scale=0.4];
  \tikzstyle{xstarB}=[star,star points=10,draw,scale=0.4];
\tikzstyle{xedge}=[thick,-];
  \tikzstyle{xedgex}=[thick,-,dashed];
  \tikzstyle{xedgedot}=[thick,-,dotted];
  
  \tikzstyle{xpath}=[color=blue,opacity=0.25,line cap=round,line width=6pt];
  \tikzstyle{xpathS}=[color=blue!40!white,opacity=0.3,line cap=round,line join=round,line width=5pt];
  \tikzstyle{xpathSx}=[color=green!40!black,opacity=0.3,line cap=round,line join=round,line width=5pt];
  \tikzstyle{xpathx}=[color=magenta,opacity=0.40,line cap=round,line width=6pt];
  \tikzstyle{xpathy}=[color=green,opacity=0.40,line cap=round,line width=6pt];
  
  \tikzstyle{xhili}=[circle,scale=1.25,opacity=0.25,fill,color=orange,draw];
  \tikzstyle{xhiliS}=[circle,scale=0.625,opacity=0.25,fill,color=orange,draw];
  \tikzstyle{xhiliIS}=[circle,scale=1.25,opacity=0.25,fill,color=magenta,draw];
  \tikzstyle{xxhiliS}=[rectangle,scale=0.85,opacity=0.25,fill,color=cyan,draw];
}

\newcommand{\Grid}[9]{
  \foreach\x in {0,...,#2}{
    \foreach \y in {0,...,#3}{
      \node (#1\x\y) at (#8*\x*\xr+#4*\xr,#9*\y*\yr+#5*\yr)[xnode,#6]{};
    }
  }
\ifnum#3>0
    \pgfmathsetmacro\yx{int(#3 - 1)}
    \foreach \x in {0,...,#2}
      \foreach \y [count=\yi] in {0,...,\yx}  
        \draw[#7] (#1\x\y)--(#1\x\yi) ;
  \fi
  \ifnum#2>0
  \pgfmathsetmacro\yx{int(#2 - 1)}
  \foreach \x in {0,...,#3}
    \foreach \y [count=\yi] in {0,...,\yx}  
      \draw[#7] (#1\y\x)--(#1\yi\x) ;
  \fi
}

\clearpage{}

\newcommand{\thetitle}{Bipartite Temporal Graphs and the Parameterized Complexity of Multistage 2-Coloring}

\date{}

\title{\thetitle}

\author{Till Fluschnik\thanks{Supported by the DFG, project MATE (NI 369/19).} }

\author{Pascal Kunz\thanks{Supported by the DFG Research Training Group 2434 ``Facets of Complexity''.}}

\affil{Technische Universit\"at Berlin, Algorithmics and Computational Complexity, Berlin, Germany \\ \symbol{123}till.fluschnik, p.kunz.1\symbol{125}@tu-berlin.de}

\begin{document}
\maketitle

\begin{abstract}
	We consider the algorithmic complexity of recognizing bipartite temporal graphs.
	Rather than defining these graphs solely by their underlying graph or individual layers, we define a bipartite temporal graph as one in which every layer can be $2$-colored in a way that results in few changes between any two consecutive layers.
	This approach follows the framework of multistage problems that has received a growing amount of attention in recent years.
	We investigate the complexity of recognizing these graphs.
	We show that this problem is NP-hard even if there are only two layers or if only one change is allowed between consecutive layers.
	We consider the parameterized complexity of the problem with respect to several structural graph parameters, which we transfer from the static to the temporal setting
	in three different ways.
	Finally, we consider a version of the problem in which
we only restrict the total number of changes throughout the lifetime of the graph.
	We show that this variant is fixed-parameter tractable with respect to the number of changes.
\end{abstract}

\section{Introduction}
\label{sec:intro}

Bipartite graphs form a well-studied class of static graphs.
A graph~$G=(V,E)$ is bipartite if it admits a proper 2-coloring.
A function~$f\colon V\to\{1,2\}$ is a \emph{proper 2-coloring} of~$G$
if for all edges~$\{v,w\}\in E$ it holds that~$f(v)\neq f(w)$.
In this work,
we study the question of what a bipartite \emph{temporal} graph is
and how fast we can determine whether a temporal graph is bipartite.
We approach this question through the prism of the novel program of multitstage problems.
Thus,
we consider the following decision problem:

\decprob{\mscTsc{} (\mscAcr{})}{msc}
{A temporal graph~$\TG=(V,(E_t)_{t=1}^\tau)$ and an integer~$d\in\Nzero$.}
{Are there~$f_1,\dots,f_\tau\colon V\to \{1,2\}$ such that~$f_t$ is a proper 2-coloring for~$(V,E_t)$ for every~$t\in\set{\tau}$
and $|\{v\in V\mid f_{t}(v)\neq f_{t+1}(v)\}|\leq d$ for every~$t\in\set{\tau-1}$?}

\noindent
In other words, 
$(\TG,d)$ is a \yes-instance
if~$\TG$ admits a proper $2$-coloring of each layer where only~$d$~vertices change colors between any two consecutive layers.

There have been various approaches to transferring graph classes from static to temporal graphs.
If $\calC$ is a class of static graphs, then the two most obvious ways of defining a temporal analog to $\calC$ are \begin{inparaenum}[(i)]
	\item including all temporal graphs whose underlying graph is in~$\calC$ or
	\item including all temporal graphs that have all of their layers in~$\calC$
\end{inparaenum}
(see, for instance, \cite{Fluschnik2020b}).
Most applied research that has employed a notion of bipartiteness in temporal graphs~\cite{Aleidi2020,Latapy2019,Wu2014} has defined it using the underlying graph, seeking to model relationships between two different types of entities.
This is certainly appropriate as long as the type of an entity is not itself time-varying.
Situations where entities can change their types require more sophisticated notions of bipartiteness.
With \mscAcr{}, 
we model situations where we expect few entities to change their type between any two consecutive time steps.
Later, 
in \cref{sec:global}, 
we will consider a model for settings where we expect few changes overall.

The issue with both of the aforementioned classical approaches to defining temporal graph classes is that they do not take the time component into account when deciding membership in a class.
For example, if the order of the layers is permuted arbitrarily, then this has no effect on membership in $\calC$ in either approach.
Defining bipartiteness in the manner we propose does take the temporal order of the layers into consideration.
It also leads to a hierarchy of temporal graph classes that are inclusion-wise between the two classes defined in the two aforementioned more traditional approaches:
It is easy to see that~$(\TG,0)$ is a \yes-instance for \mscAcr{} if and only if the underlying graph of $\TG$ is bipartite.
Conversely, if any layer of $\TG$ is not bipartite, then~$(\TG,d)$ is a \no-instance no matter the value of $d$.
The two main drawbacks to defining temporal bipartiteness in this way are that
\begin{inparaenum}[(i)]
	\item there is not one class of bipartite temporal graphs, but an infinite hierarchy depending on the value of $d$ and
	\item as we will show, testing for bipartiteness in this sense is computationally much harder, but we will attempt to partially remedy this by analyzing the problem's parameterized complexity for a variety of parameters.
\end{inparaenum}

\subparagraph{Related work.}
The multistage framework is still young, 
but several problems have been investigated in it, 
mostly in the last couple of years, 
including \prob{Matching}~\cite{Bampis2018,Chimani2021,Gupta2014}, 
\prob{Knapsack}~\cite{Bampis2019}, 
\prob{$s$\nobreakdash-$t$~Path}~\cite{Fluschnik2020a}, 
\prob{Vertex Cover}~\cite{Fluschnik2019}, 
\prob{Committee Election}~\cite{Bredereck2020}, 
and others~\cite{Bampis2020}.
The framework has also been extended to goals other than minimizing the number of changes in the solution between layers~\cite{Heeger2021,Kellerhals2021}.
Since these types of problems are NP-hard even in fairly restricted settings, most research has focused on their parameterized complexity and approximability.
+\mscAcr{} is most closely related to \prob{Multistage 2-SAT}~\cite{Fluschnik2021} (see~\cref{sec:prelim}).

\subparagraph{Our contributions.}
\begin{figure}[t]\centering
\begin{tikzpicture}

  \usetikzlibrary{patterns}
  
  \def\xr{1.4}
  \def\yr{1.4}
  \def\zr{1.75}
  \def\xsh{4.5}
  \def\xsc{1.5}
  \def\ysc{0.95}
  
  \tikzpreamble{}

  \colorlet{MyColorOne}{blue!50}
  \colorlet{colFPT}{green!33}
  \colorlet{colPK}{green!50}
  \colorlet{colXP}{cyan!33}
  \colorlet{colWh}{orange!33}
  \colorlet{colpNPh}{red!33}
  \colorlet{colOpen}{lightgray!33}
  
  \newcommand{\lightercolor}[3]{\colorlet{#3}{#1!#2!white}
    }
    
    \newcommand{\darkercolor}[3]{\colorlet{#3}{#1!#2!black}
    }
    
    \newcommand{\tcite}[1]{\\[-4pt]{\tiny(\Cref{#1})}}

  \newcommandx{\nodebox}[9][5=,7=,9=]{
    \pgfmathsetmacro{\cubex}{1.25*\xr}
    \pgfmathsetmacro{\cubey}{0.5*\yr}
    \pgfmathsetmacro{\cubez}{1.33*\zr}
    \pgfmathsetmacro{\cubeyz}{\cubey+0.38*\cubez}
    \pgfmathsetmacro{\cubexz}{\cubex+0.38*\cubez}
    
    \begin{scope}[shift={(#1)}]
\node (#2) at (-\cubex,-\cubey)[anchor=south west,minimum width=1*\cubexz cm,minimum height=1*\cubeyz cm,draw=none]{};
    \draw[black,very thin,preaction={fill=#8}] (0,0,0) -- ++(-\cubex,0,0) -- ++(0,-\cubey,0) -- ++(\cubex,0,0) -- cycle;
    \draw[black,very thin,preaction={fill=#4!50}] (0,0,0) -- ++(0,0,-\cubez) -- ++(0,-\cubey,0) -- ++(0,0,\cubez) -- cycle;
    \draw[black,very thin,preaction={fill=#6!75}] (0,0,0) -- ++(-\cubex,0,0) -- ++(0,0,-\cubez) -- ++(\cubex,0,0) -- cycle;
    \node at (-\cubex/2,-\cubey/2,0)[font=\small,align=center]{$#3_{\undp}$#9};
    \node at (-\cubex/2,0,-\cubez/2)[font=\small,align=center]{$#3_{\sump}$#7};
    \node at (0,-\cubey/2,-\cubez/2)[rotate=45,font=\footnotesize,align=center]{$#3_{\maxp}$#5};
    \end{scope}
  }
  
  \begin{scope} \nodebox{-2*\xsc,4*\ysc}{dcc}{\mathrm{dcc}}{colpNPh}[$^\dagger$]{colFPT}[\tcite{thm:fpt-dcc-sum}]{colOpen}
    \nodebox{-2*\xsc,0*\ysc}{compdiam}{\mathrm{cdi}}{colpNPh}{colpNPh}{colpNPh}
    \nodebox{4*\xsc,2*\ysc}{delta}{\Delta}{colpNPh}{colpNPh}{colpNPh}
    \nodebox{2*\xsc,0*\ysc}{deg}{\mathrm{dgn}}{colpNPh}{colpNPh}{colpNPh}
    \nodebox{0*\xsc,6*\ysc}{vc}{\mathrm{vc}}{colpNPh}[$^\dagger$]{colFPT}{colFPT}
    \nodebox{-4*\xsc,2*\ysc}{dom}{\mathrm{dom}}{colFPT}{colFPT}{colpNPh}[\tcite{prop:pnphard-dom}]
    \nodebox{-2*\xsc,6*\ysc}{dclique}{\mathrm{dcl}}{colFPT}[$^\ddagger$]{colFPT}{colFPT}
    \nodebox{-4*\xsc,4*\ysc}{is}{\mathrm{is}}{colFPT}{colFPT}{colFPT}[\tcite{prop:fpt-is}]
    \nodebox{0*\xsc,2*\ysc}{tw}{\mathrm{tw}}{colpNPh}{colpNPh}{colWh}[\tcite{prop:xp-tw}]
    \nodebox{-2*\xsc,2*\ysc}{dco}{\mathrm{dco}}{colpNPh}{colpNPh}[\tcite{prop:nphard-sump}]{colpNPh}
    \nodebox{2*\xsc,6*\ysc}{fes}{\mathrm{fes}}{colpNPh}[$^\dagger$]{colpNPh}[\tcite{prop:nphard-sump}]{colWh}[\tcite{prop:whardfesUtau}]
    \nodebox{0*\xsc,4*\ysc}{fvs}{\mathrm{fvs}}{colpNPh}{colpNPh}{colWh}
\nodebox{2*\xsc,2*\ysc}{dbi}{\mathrm{dbi}}{colpNPh}{colpNPh}{colpNPh}
    \nodebox{0*\xsc,0*\ysc}{clw}{\mathrm{clw}}{colpNPh}{colpNPh}{colpNPh}
    \nodebox{4*\xsc,6*\ysc}{bw}{\mathrm{bw}}{colpNPh}[$^\dagger$]{colpNPh}[\tcite{prop:nphard-sump-bw}]{colXP}
    \nodebox{-4*\xsc,0*\ysc}{ncc}{\mathrm{ncc}}{colFPT}[\tcite{prop:fpt-cc}]{colFPT}{colpNPh}
    
    \foreach\a\b in {deg/delta,compdiam/dco,dco/dcc,compdiam/dom,dom/is,is/dclique,dcc/dclique,
    deg/tw,
    tw/fvs,dcc/vc,fvs/vc,
    clw/dco,clw/tw,fvs/fes,
    delta/bw,tw/bw,
    ncc/dom,dbi/fvs}{\draw[->,thick,>=latex] (\a) to (\b);}
\end{scope}
  
\end{tikzpicture}
	\caption{Overview of selected structural parameters and our results (green: in \FPT;
	 orange: \XP{} and \W{1}-hard
	 red: para-\NP-hard;
	 blue: \XP{} and open whether \FPT{} or \W{1}-hard;
	 gray: open).\quad
		[$\Delta$: maximum degree;
		bw: bandwidth;
		cdi: diameter of connected component;
		clw: clique-width;
		dbi: distance to bipartite;
		dcc: distance to co-cluster;
		dcl: distance to clique;
		dco: distance to cograph;
		dgn: degeneracy;
		dom: domination number;
		fes: feedback edge number;
		fvs: feedback vertex number;
		is: independence number;
		ncc: number of connected components;
		tw: treewidth;
		vc: vertex cover number;
		for definitions of these parameters, see \cref{sec:paramzoo} in the appendix or~\cite{Sorge2019}.]
		$^\dagger$\,(\cref{prop:pNPh-max})
		$^\ddagger$\,(no polynomial kernel unless~$\NPincoNPslashpoly$)}
	\label{fig:param-hier-results}
\end{figure}

We prove that \mscAcr{} remains \NP-hard even 
if $d=1$
or
if~$\tau=2$.
We then analyze three ways of transferring structural graph parameters to the multistage setting:
the maximum over the layers,
the sum over all layers' values,
and its value on the underlying graph times~$\tau$.
We provide several (fixed-parameter) intractability and tractability results
regarding these three notions of structural parameterizations
(see~\cref{fig:param-hier-results}).
Finally, we show that a slightly modified version of the problem in which there is no restriction on the number of changes between any two consecutive layers, but on the total number of changes throughout the lifetime of the graph, is fixed-parameter tractable with respect to the number of allowed changes.

\subparagraph{Discussion and outlook.}
We proved that \mscAcr{} is \NP-hard even if~$\tau=2$ or if~$d=1$, 
but leave open whether it is fixed-parameter tractable for the combined parameter $\tau + d$.
We introduce a framework for analyzing the parameterized complexity of multistage problems regarding structural graph parameters.
While we resolve the parameterized complexity of \mscAcr{} with respect to most of the parameters, 
two cases are left open (cf.~\cref{fig:param-hier-results}).
For instance,
we proved that~\mscAcr{} is in \XP{} when parameterized by~$\bw_{\undp}$,
but we do not know whether it is in FPT or~\W{1}-hard.
Another interesting example is \mscAcr{} parameterized by~$\dcc_{\undp}$,
for which we do not know whether it is contained in~\XP{} or para-\NP-hard.
Note that we proved fixed-parameter tractability regarding~$\dcc_{\sump}$.
Finally,
we suspect that it may also be worthwhile to investigate other multistage graph problem in our framework.

\section{Preliminaries}
\label{sec:prelim}

We denote by~$\N$ ($\Nzero$) the natural number excluding (including) zero.
We use standard terminology from graph theory~\cite{Diestel2017}
and parameterized algorithmics~\cite{CyganFKLMPPS15}.

\subparagraph{Static and temporal graphs.}

We will frequently refer to graphs as static graphs in order to avoid confusion with temporal graphs.
A static graph $G=(V,E)$ is \emph{$2$-colorable} if there is a function $f\colon V \rightarrow \{1,2\}$ such that $f(u) \neq f(v)$ for all $\{u,v\} \in E$.
It is well-known that a static graph is $2$-colorable if and only if 
it does not contain any odd cycle.
This can be checked in time~$\bigO (\abs{V}+\abs{E})$ by a simple search algorithm.

A \emph{temporal graph} $\TG=(V,(E_t)_{t=1}^\tau)$ consists of a finite vertex set $V$ and $\tau$ edge sets $E_1,\ldots,E_\tau \subseteq \binom{V}{2}$.
The \emph{underlying graph} of $\TG$ is the static graph $\TG_U \coloneqq (V, \bigcup_{t=1}^\tau E_t)$.
For $t \in \{1,\ldots,\tau\}$, the $t$-th layer of $\TG$ is also a static graph, namely $\TG_t \coloneqq (V,E_t)$.
The lifetime of $\TG$ is $\tau$, the number of layers.

If $f_1,f_2 \colon X \rightarrow Y$ are two functions that share a domain and a codomain, then $\delta(f_1,f_2) \coloneqq \abs{\{x \in X \mid f_1(x) \neq f_{2}(x)\}}$ is the number of elements of $X$ whose value under $f_1$ differs from the value under $f_2$.

\subparagraph{Preliminary results.} 
There is a connection between \mscAcr{} and the \prob{Multistage 2-SAT} problem~\cite{Fluschnik2021}, 
which implies that positive algorithmic results from the latter transfer to the former.

\begin{observation}
 There is a polynomial time algorithm
 that,
 taking an instance of \mscTsc{},
 constructs an equivalent instance of~\prob{Multistage 2-SAT}
 with~$n$ variables, $2m$ clauses,
 and~$d'=d$.
\end{observation}
\begin{proof}
 For each vertex~$v$,
 construct a variable~$x_v$.
 For each edge~$\{v,w\}$ in a layer,
 construct the clauses~$(x_v\lor x_w),(\overline{x_v}\lor \overline{x_w})$.
\end{proof}
\noindent
Results on \prob{Multistage 2-Sat}~\cite{Fluschnik2021} 
imply the following.

\begin{corollary}
 \label{cor:fromms2sat}
 \mscTsc{} is
  \begin{inparaenum}[(i)]
  \item polynomial-time solvable if~$d\in\{0,n\}$,
  \item in \XP{} regarding~$n-d$ and~$\tau+d$,
  \item \FPT{} regarding~$m+n-d$ and~$n$, and
  \item admits a polynomial kernel regarding~$m+\tau$ and~$n+\tau$.
  \end{inparaenum}
\end{corollary}

\noindent
We briefly note the following:

\begin{observation}
	Given two 2-colorable graphs~$G=(V,E)$ and~$G'=(V,E')$, and 
	two 2-colorings~$f$ of~$G$ and~$f'$ of~$G'$,
	we can determine $\delta(f,f')$ in linear time.
\end{observation}

\noindent
We can strengthen the first statement in \cref{cor:fromms2sat} with the following proposition:

\begin{proposition}
	\label{prop:p-larged}
	\mscTsc{} is polynomial-time solvable if $d \geq \frac{1}{2} n$.
\end{proposition}

\begin{proof}
	Given a temporal graph $\TG=(V,(E_t)_{t=1}^\tau)$, we compute an arbitrary 2-coloring $f_t\colon V\rightarrow \{1,2\}$ of each layer~$(V,E_t)$.
	Then, 
	for each $t \in \{2,\ldots,\tau\}$ in increasing order, 
	we check whether $f_t$ introduces too many changes relative to $f_{t-1}$.
	In other words, 
	we compute $\delta(f_t,f_{t-1})$.
	If $\delta(f_t,f_{t-1}) > \frac{1}{2} n$, 
	then consider $\tilde{f}_t \colon V \rightarrow \{1,2\}$, with $\tilde{f}_t(v) = 3- f_t(v)$, 
	the coloring that reverses all assignments of $f_t$.
	Note that $\delta(\tilde{f}_t,f_{t-1}) = |\{v\in V\mid \tilde{f}_{t}(v)\neq f_{t-1}(v)\}| = |\{v\in V\mid f_{t}(v) =  f_{t-1}(v)\}| < \frac{1}{2}n $.
	Hence, 
	we set~$f_t$ to~$\tilde{f}_t$ and continue.
\end{proof}

\noindent
Testing all sequences of functions~$f_1,\ldots,f_\tau \colon V \rightarrow \{1,2\}$ gives us the following:

\begin{observation}
	\label{obs:bruteforce}
	\mscTsc{} can be decided in time $\bigO(2^{\tau n}\cdot m)$ where $\tau$ is the lifetime, $n$ the number of vertices, and $m$ the number of time edges in a temporal graph.
\end{observation}

\section{NP-hard cases}
\label{sec:nphard}
We start by proving some complexity lower bounds for \mscTsc{}.
We will show that the problem is \NP-hard in three fairly restricted cases.
\subsection{Few changes allowed}

\begin{theorem}
	\label{thm:nphard-d-fvs}
	\mscTsc{} is \NP-hard,
	even if~$d=1$.
\end{theorem}

The reduction is from the following NP-complete~\cite{Schaefer1978} problem:

\decprob{Exact 1-in-3 SAT (X1-3SAT)}{eotsat}
{A Boolean 3-CNF formula.}
{Is there a truth assignment that sets exactly one literal to true in each clause?}

\begin{construction}
	\label{constr:nphard-d-fvs}
	\begin{figure}
		\includegraphics[width=\textwidth]{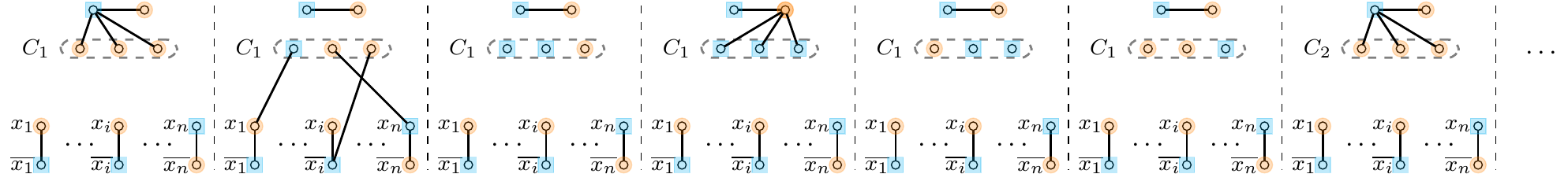}
		\caption{Illustration of \cref{constr:nphard-d-fvs}: The clause $C_1$ consists of the literals $x_1$, $\neg x_i$, and $x_n$. The assignment represented by the depicted coloring has $\alpha(x_1)=\alpha(x_i) = \top$ and $\alpha(x_n) = \bot$.}
		\label{fig:3sat-red}
	\end{figure}
	Suppose that $\varphi$ is a Boolean formula in 3-CNF over the variables~$x_1,\ldots,x_n$ with the clauses $C_1,\ldots,C_m$.
	We will construct a temporal graph $G=(V,(E_t)_{t=1}^\tau)$ with
	\begin{align*}
		V \coloneqq
		\{u_1,u_2,v_1,v_2,v_3\} \cup
		\{w_i, \bar{w_i} \mid i \in \{ 1,\ldots,n\}\}
	\end{align*}
	and $\tau\coloneqq 6m$.
	The construction is illustrated in \cref{fig:3sat-red}.
	Each clause corresponds to six layers in~$G$.
	For $j \in \{1,\ldots,m\}$, if $C_j$ consists of the literals $\ell_1,\ell_2, \ell_3$,
	then for $r \in \{1,2,3\}$ the vertex representing $\ell_r$ is $w^j_r\coloneqq w_i$ if $\ell_r = x_i$ or $w^j_r\coloneqq \bar{w_i}$ if $\ell_r = \neg x_i$.
	Then, the six layers representing~$C_j$ are:
	\begin{align*}
		E_{6j-5}  & \coloneqq \{\{u_1,u_2\},\{u_1,v_1\},\{u_1,v_2\},\{u_1,v_3\}\} \cup \{ \{w_i,\bar{w_i}\} \mid i \in \{1,\ldots,n\}\} \\
		E_{6j - 4} & \coloneqq (E_{6j-5} \setminus \{ \{u_1,v_r\} \mid r \in \{1,2,3\}\}) \cup \{ \{v_{r},w^j_r\} \mid r\in \{1,2,3\} \} \\
		E_{6j -3} & \coloneqq E_{6j-4} \setminus \{ \{v_{r},w^j_r\} \mid r\in \{1,2,3\} \} \\
		E_{6j - 2} & \coloneqq E_{6j -3} \cup \{ \{v_2,w^j_r\} \mid r \in \{1,2,3\}\} \\
		E_{6j -1} & \coloneqq E_{6j -3} \\
		E_{6j} & \coloneqq E_{6j -3}
	\end{align*}
	For a clause $C_1$ consisting of the clauses $x_1, x_n, \neg x_i$, the six layers are pictured in \cref{fig:3sat-red}.
	\cqed
\end{construction}

\begin{proof}[Proof of \cref{thm:nphard-d-fvs}]
	It is easy to see that \cref{constr:nphard-d-fvs} may be computed in polynomial time.
	We must show that $\varphi$ has a truth assignment that sets exactly one literal to true in each clause if and only if there is a multistage $2$-coloring $f_1,\ldots,f_\tau$, for $\TG$ such that $\delta(f_t,f_{t+1})\leq d \coloneqq 1$ for all~$t\in \{1,\ldots,\tau-1\}$.
	
	\RD{}
	Assume that the truth assignment $\alpha \colon \{x_1,\ldots,x_n\} \rightarrow \{\top,\bot\}$ sets exactly one literal in each clause of $\varphi$ to true.
	We will give proper 2-colorings $f_1,\ldots,f_\tau \colon V \rightarrow \{1,2\}$ of each layer of $G$.
	For all $t \in \{1,\ldots,6\tau\}$, the following colors remain the same:
	\begin{align*}
		f_t(u_1) \coloneqq 1, \quad f_t(u_2) \coloneqq 2, \quad f_t(w_i) \coloneqq \begin{cases}
			1, & \text{ if } \alpha(x_i) = \bot\\
			2, & \text{ if } \alpha(x_i) = \top
		\end{cases}
		\quad f_t(\bar{w_i}) \coloneqq 3 - f_t(w_i).
	\end{align*}
	For $j \in \{1,\ldots,m\}$, we will give the coloring $f_{6j-5},\ldots,f_{6j} \colon V \rightarrow \{1,2\}$ of the remaining vertices~$v_3$,~$v_4$, and~$v_5$ in the six layers that correspond to the clause $C_j$.
	Suppose that the vertices representing the literals $\ell_1,\ell_2,\ell_3$ in $C_j$ (in the sense described in \cref{constr:nphard-d-fvs}) are $w^j_1,w^j_2,w^j_3$.
	Exactly one of those three literals is satisfied by $\alpha$, say $\ell_s$, $s \in \{1,2,3\}$.
	Let $s' \in \{1,2,3\} \setminus \{s\}$.
	\begin{align*}
		f_{6j-5}(v_1) &= f_{6j-5}(v_2) = f_{6j-5}(v_3) \coloneqq 2, \\
		f_{6j-4}(v_s) &\coloneqq 1, \quad f_{6j-4}(v_r)\coloneqq 2 \text{ for every } r \in \{1,2,3\} \setminus \{s\}, \\
		f_{6j-3}(v_s) &= f_{6j-3}(v_{s'}) \coloneqq 1, \quad f_{6j-3}(v_r)\coloneqq 2 \text{ for every w} r \in \{1,2,3\} \setminus \{s,s'\}, \\
		f_{6j-2}(v_1) &= f_{6j-2}(v_2) = f_{6j-2}(v_3) \coloneqq 1, \\
		f_{6j-1}(v_r) &\coloneqq f_{6j-3}(v_r) \text{ for every } r \in \{1,2,3\},\\
		f_{6j}(v_r) &\coloneqq f_{6j-4}(v_r) \text{ for every } r \in \{1,2,3\}.
	\end{align*}
	It is easy to see that this coloring is proper and that only color changes in each layer.
	The changes in color are illustrated in \cref{fig:3sat-red}.
	
	\LD{}
	Suppose that $f_1,\ldots,f_\tau$ are $2$-colorings of the layers with the required properties.
	\Wilog{}, we may assume that $f_1(u_1) = 1$.
	If not, we can invert all colors.
	We define a truth assignment $\alpha \colon \{x_1,\ldots,x_\tau\} \rightarrow \{\top,\bot\}$ as follows:
	\begin{align*}
		\alpha (x_i) & \coloneqq
		\begin{cases}
			\top, & \text{ if } f_1(w_i) = 2 \\
			\bot, & \text{ if } f_1(w_i) = 1.
		\end{cases}
	\end{align*}
	We must prove that $\alpha$ satisfies exactly one literal in each clause of $\varphi$.
	We briefly note that,
	because $d=1$, 
	if two vertices are adjacent in two consecutive layers, then their colors cannot change between these two layers.
	This is because if one of the vertices is re-colored, then the other also must be, but this is not possible since at most one vertex can be re-colored from one layer to the next.
	This implies that only the colors of $v_1,v_2,v_3$ can change.
	
	Let~$j\in \{1,\ldots,m\}$.
	We must show that $\alpha$ satisfies exactly one literal in~$C_j$.
	Since $f_{6j-5}(u_1) = 1$ (as we noted, the color of $u_1$ cannot change), it follows that $f_{6j-5}(v_r) = 2$ for every~$r \in \{1,2,3\}$.
	Similarly, 
	since $f_{6j-2}(u_2) = 2$, 
	it follows that $f_{6j-2}(v_r) = 1$ for every~$r \in \{1,2,3\}$.
	Hence, 
	between the layers $6j-5$ and $6j-2$ all three vertices $v_1,v_2,v_3$ change colors,
	implying that exactly one of these vertices must change colors in each step.
	Let $v_s$, 
	$s\in\{1,2,3\}$, 
	be the vertex that changes its color to~$1$ in layer~$6j-4$.
	Let $w^j_1,w^j_2,w^j_3$ be the vertices corresponding to the literals in~$C_j$ (again in the sense described in \cref{constr:nphard-d-fvs}).
	Since $f_{6j-4}(v_s)= 1$, 
	$f_{6j-4}(v_r)= 2$ for every~$r\in\{1,2,3\} \setminus \{s\}$, 
	and the colors of~$w_i$ and $\bar{w_i}$ cannot change, 
	it follows that $\alpha$ satisfies exactly one of the literals in~$C_j$.
\end{proof}

\subsection{Few stages}

\begin{theorem}
	\label{thm:nphard-smalltau}
	\mscTsc{} is \NP-hard even on temporal graphs with $\tau=2$.
\end{theorem}

\noindent
To prove \cref{thm:nphard-smalltau},
we give a polynomial-time many-one reduction from the \NP-complete~\cite{Yannakakis1978} \prob{Edge Bipartization} problem
defined by:

\decprob{\prob{Edge Bipartization}}{gbip}
{An undirected graph $G=(V,E)$ and $k \in \Nzero$.}
{Is there a set of edges $E' \subseteq E$ with $\abs{E'} \leq k$ such that $G-E'$ is bipartite?}

\begin{construction}
	\label{constr:maxcut-red}
	Let $G=(V,E)$ be a graph and let $k\in \Nzero$.
	We assume that~$V=\{v_1,\ldots,v_n\}$.
	We construct an instance $(\TG,d)$ of \mscAcr{}
	with~$\TG\coloneqq (V',E_1,E_2)$ and~$d\coloneqq k$ as follows
	(see~\cref{fig:bip-red} for an illustrative example).
\begin{figure}
		\includegraphics[width=\textwidth]{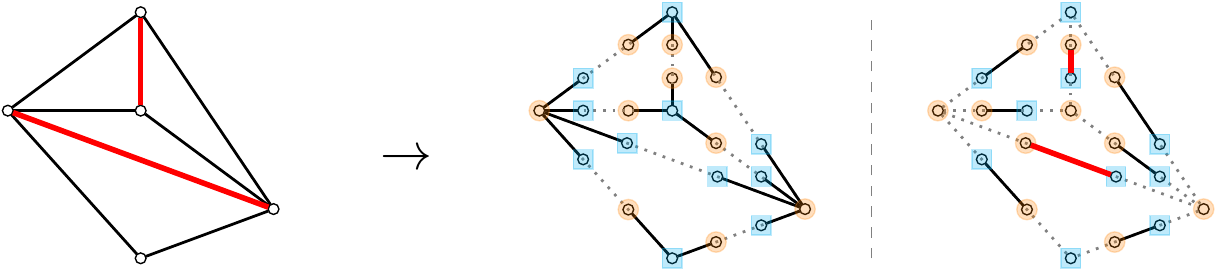}
		\caption{Illustration of \cref{constr:maxcut-red}: 
		The input graph~$G$ on the left hand-side 
		(thick/red edges indicate a solution) 
		and the output temporal graph~$\TG$ on the right-hand side 
		(thick/red edges in the second layer indicate where a recoloring was made;
		gray/dotted lines help to match with original edges from~$G$).}
		\label{fig:bip-red}
	\end{figure}

	The underlying graph of $\TG$ is obtained by subdividing each edge in $G$ twice.
	Let $u^e_i$ and $u^e_j$ be the two vertices obtained by subdividing $e=\{v_i,v_j\}$ where $u^e_i$ is adjacent to~$v_i$ and $u^e_j$ to~$v_j$.
	Then,~$V'\coloneqq V \cup \{u^e_i,u^e_j \mid e=\{v_i,v_j\} \in E\}$.
	The first layer of $\TG$ has edge set~$E_1 \coloneqq \{ \{v_i,u^e_i\} \mid i \in \{1,\ldots,n\},e\in E, v_i \in e\}$.
	The second layer has edge set~$E_2 \coloneqq \{\{u^e_i,u^e_j\} \mid e=\{v_i,v_j\} \in E \}$.
	\cqed
\end{construction}

\noindent
Next, 
we will prove the correctness of \cref{constr:maxcut-red}.

\begin{lemma}
	\label{lemma:maxcut-red}
	Instance $(G,k)$ is a \yes-instance for \prob{Edge Bipartization} if and only if instance $(\TG,d)$ output by \cref{constr:maxcut-red} is a \yes-instance for \mscTsc{}.
\end{lemma}

\begin{proof}
	\RD{} 
	Assume that $(G,k)$ is a \yes-instance and that $E'\subseteq E$ is a set of edges of size at most~$k$ such that $G-E'$ is bipartite.
	Hence, there is a proper $2$-coloring $f_0\colon V \rightarrow \{1,2\}$ of $G-E'$.
	We obtain a $2$-coloring $f_1\colon V' \rightarrow \{1,2\}$ of the first layer of $\TG$ by setting $f_1(v_i) \coloneqq f_0(v_i)$ for all~$i\in \{1,\ldots,n\}$ and $f_1(u^e_i) \coloneqq 3 - f_1(v_i)$ for all $e=\{v_i,v_j\} \in E$.
	It is easy to verify that this coloring is proper.
	A proper $2$-coloring $f_2\colon V' \rightarrow \{1,2\}$ may be defined by $f_2(v_i) \coloneqq f_1(v_i)$ and for any $e=\{v_i,v_j\}\in E$ we set
	\begin{align*}
	f_2(u^e_i) \coloneqq
	\begin{cases}
	f_1(u^e_i), & \text{ if } i<j \text{ or } f_1(u^e_i) \neq f_1(u^e_j),\\
	3 - f_1(u^e_i), & \text{ if } i>j \text{ and } f_1(u^e_i) = f_1(u^e_j).
	\end{cases}
	\end{align*}
	The only vertices that change colors between $f_1$ and $f_2$ are $u^e_i$ with $e=\{v_i,v_j\} \in E$, $i >j$, and~$f_1(v_i) = f_1(v_j)$.
	However, $f_1(v_i) = f_1(v_j)$ implies that $f_0(v_i) = f_0(v_j)$ and hence $e \in E'$.
	Since~$\abs{E'} \leq k = d$, it follows that at most $d$ vertices change colors.
	
	\LD{}
	Suppose that $f_1,f_2\colon V' \rightarrow \{1,2\}$ are proper $2$-colorings of $\TG_1$ and $\TG_2$, respectively.
	Let $E'\coloneqq \{e=\{v_i,v_j\} \in E \mid f_1(u^e_i) = f_1(u^e_j) \}$.
	Since $\{u^e_i,u^e_j\} \in E_2$, it follows that one of the vertices $u^e_i,u^e_j$ must change colors between $f_1$ and $f_2$ if $e \in E'$.
	This implies that $\abs{E'} \leq d = k$.
	For $e = \{v_i,v_j\} \in E \setminus E'$, it follows that $f_1(u^e_i) \neq f_1(u^e_j)$ and hence $f_1(v_i) \neq f_1(v_j)$.
	This implies that the restriction of $f_1$ to $V$ induces a proper $2$-coloring of $G-E'$.
\end{proof}

\noindent
This allows us to prove \cref{thm:nphard-smalltau}.

\begin{proof}[Proof of \cref{thm:nphard-smalltau}]
	It is easy to see that \cref{constr:maxcut-red} can be computed in polynomial time.
	The claim follows by \cref{lemma:maxcut-red}.
\end{proof}

The reduction also implies the following:

\begin{proposition}
	Unless \ETHbreaks, 
	\mscTsc{} admits no $\bigO(2^{o(n+m)})$-time algorithm,
	where $n$~is the number of vertices and $m$~is the number of time edges in a temporal graph, 
	even for $\tau=2$.
\end{proposition}
\begin{proof}
	Unless \ETHbreaks, \prob{Edge Bipartization} cannot be solved in time~$\bigO(2^{o(n)})$,
	where $n$ is the number of vertices.
	This follows from the corresponding lower bound for \prob{Maximum Cut}~\cite{Okrasa2020}.
	The instance output by \cref{constr:maxcut-red} contains $n+2m$ vertices.
	The claim follows by \cref{lemma:maxcut-red}.
\end{proof}

\noindent
Next we present a further reduction to \mscAcr{}.
We will use this reduction to prove parameterized lower bounds in \cref{sec:param}.
The reduction is from the NP-complete \prob{Clique} problem.

\begin{construction}
	\label{constr:red-from-clique}
	
	\begin{figure}
		\includegraphics[width=\textwidth]{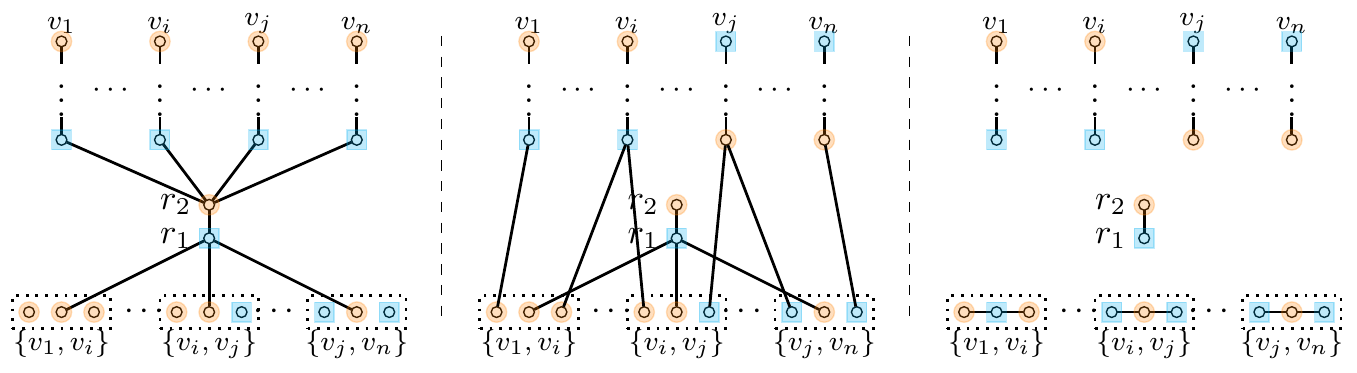}
		\caption{
		Illustration of \cref{constr:red-from-clique}: This temporal graph represents a static graph with~$E=\{\{v_1,v_i\}, \{v_i,v_j\}, \{v_j,v_n\}\}$.
			The vertices $r_3,\ldots,r_{d+1}$ are omitted from the illustration.
			Blue represents the color $1$ in the proof of \cref{thm:nphard-smalltau} and orange the color $2$.
			Then, the coloring depicted represents a clique that contains $v_j$ and $v_n$, but not $v_1$ and~$v_i$.
		}
		\label{fig:clique-red}
	\end{figure}
	Let $(G=(V,E),k)$ be an instance for \prob{Clique}.
	Let $V=\{v_1,\ldots,v_n\}$ and~$\abs{E}=m$.
	We may assume that $k\geq 3$
	and
	that $m\geq \binom{k}{2}$
	(otherwise, $(G,k)$ is clearly a \no-instance).	
	Finally, we assume that $m-\binom{k}{2}$ is divisible by~$k$.
	If it is not, 
	we can simply add a star with 
	$k-((m-\binom{k}{2})\bmod k)$
	many leaves since this does not add or remove a $k$-clique with~$k\geq 3$.
	We will construct an instance $(\TG,d)$ for \mscTsc{} consisting of a temporal graph~$\TG=(V',E_1,E_2,E_3)$ with three layers and $d \coloneqq m - \binom{k}{2}$.
	
	The general idea is that each vertex in $G$ is represented by a path,
	and its coloring in the second and third layer determine whether or not the represented vertex is in the clique.
	The restriction on the color changes between the layers ensure that at least $k$ vertices are chosen and that all pairs of chosen vertices are adjacent.
	The construction is illustrated in \cref{fig:clique-red}.
	
	We start by defining the vertex set $V'$.
	Let $\ell \coloneqq (m-\binom{k}{2})/k$.
	Then, $V_1 \coloneqq \{ u^v_i \mid v \in V, i \in \{1,\ldots,\ell\}\}$.
	Next, we let $V_2 \coloneqq \{ w^e_1, w^e_2, w^e_3 \mid e \in E\}$.
	Finally, $V_3 \coloneqq \{r_i \mid i \in \{1,\ldots,d+1\}\}$.
	Then, we define~$V'\coloneqq V_1 \cup V_2 \cup V_3$.
	
	We must now define the edge sets $E_1,E_2,E_3$.
	We start by defining a set of edges that will be present in every layer of $\TG$.
	Essentially, 
	the vertices in~$V_1$ that correspond to the same vertex in~$G$ and 
	the vertices in~$V_3$ each form a path. 
	Let $E^P\coloneqq \{\{u^v_i,u^v_{i+1}\} \mid v \in V, i \in \{1,\ldots,\ell-1\}\} \cup \{ \{ r_i,r_{i+1} \} \mid i \in \{1,\ldots,d \}\}$.
	Then:
	\begin{align*}
		E_1 & \coloneqq  E^P \cup \{ \{r_1,w^e_2\} \mid e \in E\} \cup \{ \{r_2,u^v_1\} \mid v \in V\}, \\
		E_2 & \coloneqq E^P \cup \{ \{r_1,w^e_2\} \mid e \in E\} \cup \{ \{u_1^{v_i},w^e_1\}, \{u_1^{v_j},w^e_3\} \mid e=\{v_i,v_j\}\in E, i<j \}, \\
		E_3 & \coloneqq E^P \cup \{ \{ w^e_1,w^e_2 \}, \{ w^e_2,w^e_3 \} \mid e \in E  \}. &&& \text{\cqed}
	\end{align*}
	
\end{construction}

\begin{lemma}
  \label{lem:red-from-clique}
	\cref{constr:red-from-clique} can be computed in polynomial time and the input instance is equivalent to the output instance.
\end{lemma}

\begin{proof}
	It is easy to verify that \cref{constr:red-from-clique} can be computed in polynomial time.
	We must show that $(G=(V,E),k)$ is a \yes-instance of~\prob{Clique} 
	if and only if 
	$(\TG=(V,E_1,E_2,E_3),d)$ is a \yes-instance of~\mscAcr{}.
	
	\RD{} 
	Suppose that $X \subseteq V$ is a clique of size exactly $k$ in $G$.
	We will give $f_1,f_2,f_3\colon V' \rightarrow \{1,2\}$ proving that $(\TG,d)$ is a \yes-instance.
	Let
	\begin{align*}
		f_1(u_i^v) \coloneqq
		\begin{cases}
			1, & \text{if $i$ is odd},\\
			2, & \text{if $i$ is even},
		\end{cases}
		\quad
		f_1(r_i) \coloneqq
		\begin{cases}
		1, & \text{if $i$ is odd},\\
		2, & \text{if $i$ is even}.
		\end{cases}
	\end{align*}
	For any $e = \{v_i,v_j\} \in E$, $i<j$, let:
	\begin{align*}
		f_1(w_1^e) \coloneqq
		\begin{cases}
			1, & \text{if } v_i \in X,\\
			2, & \text{if } v_i \notin X,
		\end{cases}
		\quad
		f_1(w_2^e) \coloneqq 2,
		\quad
				f_1(w_3^e) \coloneqq
		\begin{cases}
			1, & \text{if } v_j \in X,\\
			2, & \text{if } v_j \notin X.
		\end{cases}
	\end{align*}
	It is easy to see that this coloring of $(V',E_1)$ is proper.
	We continue by giving the coloring $f_2$ of the second layer.
	First, $f_2(x) \coloneqq f_1(x)$ for all $x \in V_2 \cup V_3$.
	The colors of vertices in $V_1$, however, can change.
	Let:
	\begin{align*}
		f_2(u^v_i) \coloneqq 
		\begin{cases}
			1, & \text{if $i$ is odd and $v \notin X$,}\\
			2, & \text{if $i$ is even and $v \notin X$,}\\
			2, & \text{if $i$ is odd and $v \in X$,}\\
			1, & \text{if $i$ is even and $v \in X$.}
		\end{cases}
	\end{align*}
	Again, it is easy to see that $f_2$ is a proper coloring of $(V',E_2)$.
	Note that the only vertices that change colors are $u^v_i$ with $v \in X$.
	There are exactly $\abs{X} \cdot \ell = m - \binom{k}{2} = d$ such vertices.
	We conclude by defining $f_3$.
	The colors of $V_1$ and $V_3$ do not change, so let $f_3(x)\coloneqq f_2(x)$ for all~$x \in V_1 \cup V_3$.
	Consider any $e = \{v_i,v_j\} \in E$, $i<j$.
	Then:
	\begin{align*}
		\begin{rcases*}
			f_3(w^e_1) \coloneqq 2 \\
			f_3(w^e_2) \coloneqq 1 \\
			f_3(w^e_3) \coloneqq 2
		\end{rcases*}
		\text{ if } v_i,v_j \notin X \quad \text{ and } \quad
		\begin{rcases*}
			f_3(w^e_1) \coloneqq 1 \\
			f_3(w^e_2) \coloneqq 2 \\
			f_3(w^e_3) \coloneqq 1
		\end{rcases*}
		\text{ if } v_i \in X \text{ or } v_j \in X
	\end{align*}
	It is also not difficult to see that $f_3$ is a proper coloring of $(V',E_3)$.
	To see that exactly $d$ vertices change colors, 
	first note that only vertices in $V_2$ change colors.
	Moreover, 
	if $v_i$ and $v_j$ are both in the clique $X$, 
	then none of the vertices $w^e_1,w^e_2,w^e_3$ change colors.
	However, 
	if one of $v_i$ and $v_j$ is not in $X$, 
	then exactly one of those three vertices changes colors.
	Hence, 
	the number of changes to the coloring is $m-\binom{k}{2} = d$.
	
	\LD{}
	Suppose that $f_1,f_2,f_3 \colon V' \rightarrow \{1,2\}$ are proper colorings of the layers of~$\TG$ such that only~$d$~vertices change colors between any two consecutive layers.
	\Wilog, we may assume that $f_1(r_1) = 1$.
	Since the vertices in $V_3$ form a path in every layer, all of these vertices must be re-colored if any one of them is.
	Since there are $d+1$ such vertices, 
	their color cannot be changed.
	Hence, 
	we assume that for every~$t\in \{1,2,3\}$ it is the case that $f_t(r_i) = 1$ if $i$ is odd and $f_t(r_i) = 2$ if $i$ is even.
	This directly implies that for all $v\in V$,
	$f_1(u^v_i) = 1$ if $i$ is odd and $f_1(u^v_i) = 2$ if $i$ is even .
	Moreover, 
	by the same reasoning, 
	we conclude that $f_1(w^e_2) = f_2(w^e_2) = 2$ for every~$e \in E$.
	Note that $w^e_1$ and $w^e_3$ are both isolated in~$(V',E_1)$.
	Hence, 
	their colors in the first layer are irrelevant and we may assume that their color does not change between the first two layers.
	All in all, 
	it follows that the only vertices that change color between $E_1$ and $E_2$ are in~$V_1$.
	However, 
	because $u^v_1,\ldots,u^v_\ell$ form a path, 
	we conclude that if $u^v_i$ changes colors for some~$i\in\set{\ell}$, 
	then $u^v_j$ changes colors for every~$j \in \{1,\ldots,\ell\}$.
	Let $X \coloneqq \{ v \in V \mid f_1(u^v_1) \neq f_2(u^v_1)\}$.
	We also note that $\abs{X} \leq \frac{d}{\ell} = k$,
	
	It remains to show that $X$ is a clique in~$G$ and that $\abs{X} \geq k$.
	To this end, 
	first note that for any $e=\{v_i,v_j\} \in E$, 
	$i<j$, 
	it is the case that $f_2(w^e_1) = 1$ if and only if $v_i \in X$ and $f_2(w^e_3) = 1$ if and only if $v_j \in X$, while $f_2(w^e_2) = 2$.
	On the other hand, in the final layer, the edges between these vertices imply that~$f_3(w_1^e) \neq f_3(w_2^e)$ and $f_3(w_2^e)\neq f_3(w_3^e)$.
	Hence, the only way that $w_1^e,w_2^e,w_3^e$ can all keep their colors is if $v_i,v_j \in X$.
	Hence, the number of edges that do not have both endpoints in~$X$ is at most $d$.
	Therefore, $m - \binom{\abs{X}}{2} \leq d = m - \binom{k}{2}$, implying that $\abs{X} \geq k$.
	Since $\abs{X} \leq k$ as we noted above, this forces $\abs{X} = k$.
	Because the number of edges with both endpoints in $X$ is at least $m - d = \binom{k}{2} = \binom{\abs{X}}{2}$, this leads us to conclude that $X$ is a clique.
\end{proof}

\subsection{Few edges per layer}

\begin{theorem}
	\label{thm:nphard-fewedges}
	\mscTsc{} is NP-hard even for $d=1$ and restricted to temporal graphs where each layer contains just three edges and has maximum degree one.
\end{theorem}

\noindent
We will prove this using a reduction from \mscAcr{} with $d=1$, which is NP-hard by \cref{thm:nphard-d-fvs}.

\begin{construction}
	\label{constr:nphard-fewedges}
	\begin{figure}[t]
		\includegraphics[width=\textwidth]{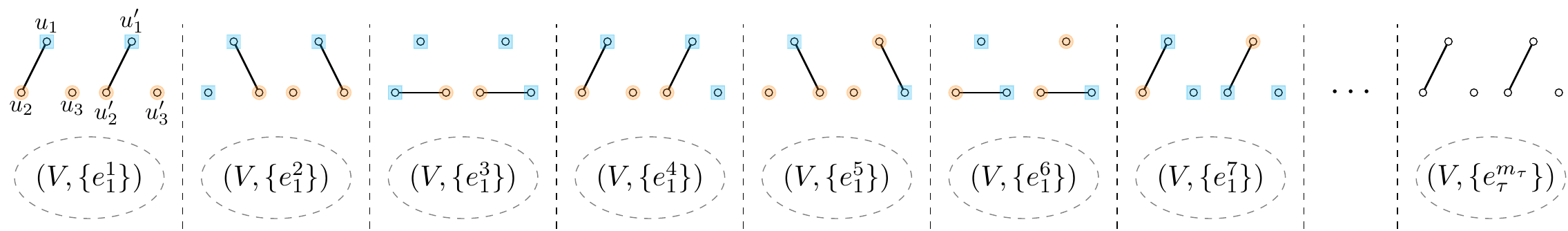}
		\caption{Illustration to \cref{constr:nphard-fewedges}.}
		\label{fig:nphard-fewedges}
	\end{figure}
	
	Let $(\TG=(V,(E_t)_{t=1}^\tau),d=1)$ be an instance for \mscTsc{}
	where for every~$t,t'\in\set{\tau}$ it holds that~$|E_t|=|E_{t'}|$,
	$|E_t|\geq 4$,
	and~$|E_t|\bmod 3=1$
	(we can guarantee this by adding a star~$K_{1,q}$, 
	$q=\binom{|V|}{2}$,
	to the underlying graph and add edges from the star to layers to fulfill the criterias).
	We will construct an instance~$(\TG',d)$ as required
	(see~\cref{fig:nphard-fewedges} for an illustration).
	The general idea is that we spread the edges of any one layer in $\TG$ to several layers in $\TG'$ by presenting the edges one at a time.
	In order to ensure that the solution does not change between the layers in~$\TG'$ corresponding to the same layer in $\TG$, 
	we additionally introduce a gadget that uses up the budget for changes to the solution.
	
	For every~$t\in\set{\tau}$,
	let $m\coloneqq \abs{E_t}$ and order the edges in~$E_t$ arbitrarily as~$e_t^1,\ldots,e_t^{m}$.
	Let~$V' \coloneqq V \uplus V^+$ with $V^+\coloneqq \{u_1,u_2,u_3,u'_1,u'_2,u'_3\}$.
	The six new vertices will be used to implement the aforementioned gadget.
  Let
	\begin{align*}
	E^+_1 \coloneqq \{\{u_1,u_2\},\{u'_1,u'_2\}\}, \quad
	E^+_2 \coloneqq\{\{u_1,u_3\},\{u'_1,u'_3\}\}, \text{ and} \quad
	E^+_3 \coloneqq\{\{u_2,u_3\},\{u'_2,u'_3\}\}.
	\end{align*}
	For every~$t\in\set{\tau}$ and~$k\in\set{m}$,
	let
	$E_t^k\coloneqq \{e_t^k\} \cup E^+_{k'}$ where $k'\coloneqq (k-1 \bmod 3) + 1$.
	Note that for every~$t\in\set{\tau-1}$,
	it holds that~$E^+_1\subseteq E_t^{m}\cap E_{t+1}^1$ 
	(recall that~$1\leq (m-1)/3\in\N$).
	Let~$\tau' \coloneqq \tau\cdot m$.
	Then, 
	the output instance is $(\TG',d)$ with
$\TG'\coloneqq(V',(E_t')_{t=1}^{\tau'})$
where~$E_{(p-1)\cdot m+k}'\ceq E_p^k$ for~$p\in\set{\tau}$ and~$k\in\set{m}$.
	\cqed
\end{construction}

\begin{proof}[Proof of \cref{thm:nphard-fewedges}]
	It is easy to see that \cref{constr:nphard-fewedges} can be computed in polynomial time and that any instance that it outputs has the properties in the statement of the theorem.
	We must still show that $(\TG,d)$ is a \yes-instance if and only if $(\TG',d)$ is.
	
	\RD{}
	Suppose that $f_1,\ldots,f_\tau$ are 2-colorings of the layers of $\TG$ such that $\delta(f_t,f_{t+1})\leq 1$ for every~$t\in\set{\tau-1}$.
	Then consider $f^k_t \colon V' \rightarrow \{1,2\}$ for $t \in \{1,\ldots,\tau\}$ and $k \in \{1,\ldots m\}$ defined by
	\begin{align*}
		f^k_t(v) \coloneqq
		\begin{cases}
			f_t(v), & \text{ if } v \in V, \\
			g_t^k(v), & \text{ if } v \in V^+,
		\end{cases}
	\end{align*}
	where $g_t^k$, the coloring pictured in \cref{fig:nphard-fewedges}, is obtained in the following manner.
	We use an arbitrary 2-coloring of $(V^+,E_1^+\cup E_2^+)$ for $g^1_1$.
	Moreover, for $t>1$ we let $g_t^1 \coloneqq g_{t-1}^{m}$.
	If $2\leq k<m$ is even, 
	then we obtain~$g^{k+1}_t$ from $g^{k}_t$ by changing the color of one of the vertices $\{u_1,u_2,u_3\}$ such that 
	$g^{k+1}_t$ properly colors~$\{u_1,u_2,u_3\}$ with respect to both $E_{k+1}^+$ and $E_{k+2}^+$.
	If $3\leq k<m$ is odd, 
	we do the same for~$\{u'_1,u'_2,u'_3\}$.
	
	We must show that this coloring is proper and that there is at most one change between any two consecutive colorings.
	No edge $e_t^k$ from $\TG$ is monochromatic because $f_1,\ldots,f_\tau$ are proper colorings by assumption.
	If $t=1$ and $k=1$ or if $k>1$, then $g^k_t$ properly colors $V^+$ by construction.
	If~$t>1$ and $k=1$, then $g_t^1 = g_{t-1}^{m}$.
	Recall that $E^+_1 \subseteq E_t^1\cap E_{t-1}^{m}$.
	Thus,
	$g_t^1$ also properly colors~$V^+$.
	It remains to show that there is at most one change between any two consecutive colorings.
	If $k<m$, then the colors of the vertices in $V$ do not change between the stages $E_t^k$ and $E_t^{k+1}$, 
	while only one of the vertices in $V^+$ changes colors by the construction of $g^k_t$.
	Between the stages $E_t^{m}$ and $E_{t+1}^{1}$ only one vertex in~$V$ changes colors by assumption, 
	while the vertices in~$V^+$ are not re-colored.
	
	\LD{}
	Suppose that $f_t^k$ for $t\in\{1,\ldots,\tau\}$ and $k \in \{1,\ldots,m\}$ are proper colorings of~$(V',E_t^k)$ such at most one vertex changes colors between $E_t^k$ and $E_t^{k+1}$ or between $E_t^{m}$ and~$E_{t+1}^1$.
	
	First, we claim that, if $k<m$, then at least one vertex in $V^+$ must change colors between~$E_t^k$ and~$E_t^{k+1}$.
	If $k < m$, consider the graph $(V,E_t^k \cup E_t^{k+1} \cup E_t^{k+2})$.
	The vertices $\{u_1,u_2,u_3\}$ and~$\{u'_1,u'_2,u'_3\}$ each induce a~$K_3$ in this graph.
	Hence, at least two changes must be made between $E_t^k$ and $E_t^{k+2}$.
	Since only one change can be made in each step, this implies that one must be made in each.
	
	This claim implies that the vertices in $V$ do change colors under $f_t^k$, except possibly between~$E_t^{m}$ and~$E_{t+1}^1$.
	We define $f_1,\ldots,f_\tau \colon V \rightarrow \{1,2\}$ by $f_t(v) \coloneqq f_t^1(v) = \ldots = f_t^{m}(v)$. Note that $f_t$ is a proper coloring of $(V,E_t)$ because $f_t^1,\ldots,f_t^{m}$ are proper colorings.
	Moreover, only one vertex changes colors between $f_t$ and $f_{t+1}$.
\end{proof}

\section{Parameterized complexity}
\label{sec:param}

In the previous section we showed that \mscTsc{} is NP-hard, 
even for constant values of $\tau$ and $d$.
In this section,
we study the parameterized complexity of~\mscTsc{}.
To begin with,
we will now show that \mscTsc{} is \fpt{} with respect to $n-d$.
This is in contrast to \prob{Multistage 2-SAT}, 
which is \W{1}-hard with respect to this parameter~\cite[Theorem~3.6]{Fluschnik2021}.

\begin{proposition}
	\label{prop:fpt-n-d}
	\mscTsc{} is \fpt{} with respect to $n-d$.
\end{proposition}
\begin{proof}
	If $d \geq \frac{n}{2}$, the problem can be solved in polynomial time (see \cref{prop:p-larged}).
	If $d < \frac{n}{2}$, then it follows that $n < 2(n-d)$.
	Hence, the fixed-parameter tractability of \mscAcr{} with respect to $n$ (see \cref{cor:fromms2sat}) implies fixed-parameter tractability with respect to $n-d$.
\end{proof}

Additionally, we note the following kernelization lower bound.
\begin{proposition}
	\label{prop:no-pk}
	Unless~\NPincoNPslashpoly,
	\mscTsc{} admits no problem kernel 
of size polynomial in the number~$n$ of vertices.
\end{proposition}

\begin{proof}
	We give an AND-composition~\cite{BodlaenderDFH09} from \mscAcr{} into~\mscAcr{} parameterized by~$n$,
	which then yields the theorem's statement~\cite{Drucker15}.
	Let~$\I_1=(\TG^1,d),\dots,\I_p=(\TG^p,d)$ be~$p$ instances of~\mscAcr{} with~$d=1$.
	Note that we can assume~\cite{BodlaenderJK14}
	that~$V$ denotes the vertex set
	and $(E_1^q,\dots,E_\tau^q)$ denotes the edge sequence for each~$\TG^q$,
	$q\in\set{p}$.
	Let~$n\ceq |V|$.
	We build the temporal graph~$\TG$ with vertex set~$V$ and sequence
	\[ (E_1^1,\dots,E_\tau^1,E_1^{1,2},\dots,E_n^{1,2},E_1^2,\dots,E_\tau^2,E_1^{2,3},\dots,E_n^{2,3},E_1^3,\dots,\dots,\dots,E_\tau^p), \]
	where~$E_i^{q,q+1}\ceq \emptyset$ for every~$q\in\set{p-1}$ and~$i\in\set{n}$.
	We claim that~$\I\ceq (\TG,d)$ is a \yes-instance
	if and only if
	$\I_q$ is a \yes-instance for every~$q\in\set{p}$. 
	Since the forward direction is immediate,
	we only discuss the backward direction in the following.
	
	For each~$q\in\set{p}$,
	let~$f_1^q,\dots,f_\tau^q\colon V\to\{1,2\}$ be proper colorings of~$(V,E_1^q),\dots,(V,E_\tau^q)$
	such that~$\delta(f_i^q,f_{i+1}^q)\leq d$ for every~$i\in\set{\tau-1}$.
	Note that for every~$q\in\set{p-1}$,
	since there are~$n$ empty layers between~$(V,E_\tau^q)$ and~$(V,E_1^{q+1})$,
	we can get from~$f_\tau^q$ to~$f_1^{q+1}$ in at most~$n$ steps
	with having consecutive coloring not differ in more than one vertex.
	This way,
	we can obtain a solution to~$\I$,
	witnessing that~$\I$ is a \yes-instance.
\end{proof}

In the following, we will consider the parameterized complexity of \mscTsc{} with respect to structural graph parameters.
Research on the parameterized complexity of multistage problems has thus far mostly focused on the parameters that are given as part of the input such as $d$ or~$\tau$.
Although Fluschnik~et~al.~\cite{Fluschnik2020a} considered the vertex cover number and maximum degree of the underlying graph, there has been no systematic study of multistage problems concerned with structural parameters of the input temporal graph.
We seek to initiate this line of research in the following.
It follows the call by Fellows~et~al.~\cite{Fellows2009,Fellows2013} to investigate problems' ``parameter ecology'' in order to fully understand what makes them computationally hard.
We will begin with a short discussion of how graph parameters can be applied to multistage problems.
This question is closely related to issues that arise when applying such parameters to temporal graph problems (see \cite{Fluschnik2020} and \cite[Sect.~2.4]{Molter2020}).

A \emph{(temporal) graph parameter} $p$ is a function that maps any (temporal) graph $G$ to a nonnegative integer $p(G)$.
We will consider three ways of transferring graph parameters to temporal graphs.
If $p$ is a graph parameter, $\TG=(V,(E_t)_{t=1}^\tau)$ is a temporal graph, $\TG_t \coloneqq (V,E_t)$ its $t$-th layer, and $\TG_U \coloneqq (V,\bigcup_{t=1}^\tau E_t)$ its underlying graph, then we define:
\begin{align*}
	p_{\maxp}(\TG) & \coloneqq \max_{t \in \{1,\ldots,\tau\}} p(\TG_t), && \text{(\emph{maximum} parameterization)} \\
	p_\sump(\TG) & \coloneqq \sum_{t=1} ^\tau \max\{1,p(\TG_t)\} \text{, and} && \text{(\emph{sum} parameterization)} \\
	p_{\undp}(\TG) & \coloneqq p(\TG_U) + \tau. && \text{(\emph{underlying graph} parameterization)}
\end{align*}

We will briefly explain our choice to define these parameters in this manner and describe the relationship between the parameters.
For any two (temporal) graph parameters~$p_1$ and $p_2$, the first parameter~$p_1$ \emph{is larger than}~$p_2$, written~$p_1 \succeq p_2$ or~$p_2 \preceq p_1$, if there is a function $f\colon \Nzero \rightarrow \Nzero$ such that~$f(p_1(G)) \geq p_2(G)$ for all (temporal) graphs $G$.
Such relationships between parameters are useful because, 
if $p_1\succeq p_2$, 
then any problem that is fixed-parameter tractable with respect to~$p_2$ is also fixed-parameter tractable with respect to~$p_1$.
The $\succeq$-relation between static graph parameters is well-understood~\cite{Jansen2013,Sasak2010,Schroder2019,Sorge2019,Vatshelle2012}.
We will use these relationships implicitly and explicitly throughout this article.
Many of the results claimed in \cref{fig:param-hier-results} will not be explicitly proved,
because they are immediate consequences of other results and the $\succeq$-relation.
The relationships under $\succeq$ between selected graph parameters are pictured in that figure.

When it comes to transferring graph parameters from the static to the multistage setting, 
the parameters~$p_{\maxp}$ and~$p_\undp$ simply apply the graph parameter to the individual layers and to the underlying graph, 
respectively, 
and were used in a similar manner by Fluschnik et~al.~\cite{Fluschnik2020} and Molter~\cite{Molter2020}.
The reasoning behind the definition of the sum parameterization may not be quite as obvious.
It seems natural to consider the sum of the parameters over all layers.
The issue with this is that it may not preserve the $\succeq$-relation.
For example, 
it is well-known that feedback vertex number is a larger parameter 
(in the sense of~$\succeq$) 
than treewidth.
However, 
consider a temporal graph where each layer is a forest.
Then, the sum of the feedback vertex numbers of the layers is~$0$, 
but the sum of the layers' treewidths is~$\tau$.
Hence, 
treewidth is no longer bounded from above by the feedback vertex number.
Our definition gets around this problem.
In fact, 
all three aforementioned ways of transferring parameters from the static to the multistage setting preserve the $\succeq$-relation:

\begin{proposition}
	\label{prop:param-preservation}
	Let $p$ and $q$ be graph parameters with $p \succeq q$.
	Then, $p_{\alpha} \succeq q_{\alpha}$ for any $\alpha \in \{\maxp, \sump, \undp\}$.
\end{proposition}
\begin{proof}
	Let $f \colon \Nzero \rightarrow \Nzero$ be a function such that $f(p(G)) \geq q(G)$ for all static graphs $G$.
	Without loss of generality, 
	we may assume that 
	\begin{inparaenum}[(i)]
    \item $f$~is monotonically increasing, 
          that is,
          $f(a) \geq f(b)$ if~$a\geq b$,
          and
    \item $f(a)\geq a$ for every~$a\in\Nzero$ 
    (consider~$f'(a)\ceq a+\max_{b\in\set{a}}f(b)$, $a\in\Nzero$, for instance).
	\end{inparaenum}
	
	Let $\TG$ be an arbitrary temporal graph.
	Then:
	\begin{align*}
		f(p_{\maxp}(\TG)) 
		&= f\left(\max_{t \in \{1,\ldots,\tau\}} p(\TG_t)\right) 
		\stackrel{\text{(i)}}{=} \max_{t \in \{1,\ldots,\tau\}} f(p(\TG_t)) 
		\geq \max_{t \in \{1,\ldots,\tau\}} q(\TG_t) 
		= q_{\maxp}(\TG)
	\end{align*}
	For $n \in \N$, let $\Part(n)$ denote the set of all partitions of $n$, 
	that is all possible ways of writing~$n$ as $n = n_1 + n_2 + \ldots + n_r$ for $r\geq 1$ and $n_i \in \N$.
	Let $g\colon \Nzero \rightarrow \Nzero$ with:
	\begin{align*}
		g(0)\coloneqq 0, \quad g(n) & \coloneqq \max \left\{ \sum_{i=1}^r f(n_i) \Bigm\vert (n_1,\ldots,n_r) \in \Part(n) \right\} \text{ if } n>0.
	\end{align*}
	The maximum is well-defined, because $\Part(n)$ is finite.
For any temporal graph~$\TG$,
  we have:
	\begin{align*}
		g(p_{\sump}(\TG)) &= g \left( \sum_{t=1}^\tau \max\{1,p(\TG_t)\}  \right) \geq \sum_{t=1}^\tau f(\max\{1,p(\TG_t)\}) \stackrel{\text{(i)}}{=} \sum_{t=1}^\tau \max\{f(1),f(p(\TG_t))\} \\
		&\stackrel{\text{(ii)}}{\geq} \sum_{t=1}^\tau \max\{1,q(\TG_t)\} = q_{\sump}(\TG).
	\end{align*}
	(Note that the first inequality relies on the fact that every term in the sum is at least $1$, since a partition can only be composed of positive summands.
	Therefore, this argument would not apply, if we defined the sum parameterization as simply the sum over the parameters of the individual layers.) 
	
Lastly,
  for any temporal graph~$\TG$,
  we have:
	\begin{align*}
		g(p_{\undp}(\TG)) &= g(p(\TG_U) + \tau) \geq f(p(\TG_u)) + f(\tau) \stackrel{\text{(ii)}}{\geq} q(\TG_U) + \tau = q_\undp(\TG).
		\qedhere
	\end{align*}
\end{proof}

\noindent
Finally, we will briefly consider the relationship between $p_{\maxp}$, $p_{\sump}$, and $p_\undp$.
We will say that a graph parameter $p$ is \emph{monotonically increasing} if for any two static graphs $G=(V,E)$ and~$G'=(V,E')$ with the same vertex set, it is the case that $E\subseteq E'$ implies $p(G) \leq p(G')$.
Conversely, it is \emph{monotonically decreasing} if $E\subseteq E'$ implies $p(G) \geq p(G')$.

\begin{proposition}
	\label{prop:param-relationships}
	Let $p$ be a graph parameter.
	Then:
	\begin{enumerate}[(i)]
		\item $p_{\maxp} \preceq p_{\sump}$,
		\item $p_{\sump} \preceq p_{\undp}$, if $p$ is monotonically increasing, and
		\item $p_{\sump} \succeq p_{\undp}$, if $p$ is monotonically decreasing.
	\end{enumerate}
\end{proposition}
\begin{proof}
	\begin{enumerate}[(i)]
		\item Obvious.
		\item Let $\TG$ be a temporal graph.
		Note that since $\TG_t \subseteq \TG_U$, 
		it follows that~$p(\TG_t) \leq p(\TG_U)$ for all~$t\in\{1,\ldots,\tau\}$.
		Hence:
		\begin{align*}
			p_{\sump} (\TG) & = \sum_{t=1}^\tau \max\{1,p(\TG_t)\} \leq \tau + \sum_{t=1}^\tau p(\TG_t) \leq \tau + \tau \cdot p(\TG_U) \leq (\tau + p(\TG_U))^2 = p_\undp(\TG)^2.
		\end{align*}
		\item Let $\TG$ be a temporal graph.
		Note that since $\TG_t \subseteq \TG_U$, 
		it follows that~$p(\TG_t) \geq p(\TG_U)$ for all~$t\in\{1,\ldots,\tau\}$.
		If $\tau = 1$ or $p(\TG_U) \leq 1$, the claim is obvious.
		Otherwise, we have that:
		\begin{align*}
			p_{\sump} (\TG) = \sum_{t=1}^\tau \max\{1,p(\TG_t)\} \geq \sum_{t=1}^\tau \max\{1,p(\TG_U)\} \geq \sum_{t=1}^\tau p(\TG_U) = \tau \cdot p(\TG_U) \geq p_\undp(\TG).
		\end{align*}
	\end{enumerate}
\end{proof}

\noindent
We will now investigate the problem's parameterized complexity with respect to the three types of parameterizations.
\cref{fig:param-hier-results} gives an overview of our results and of the abbreviations we use for the parameters.
Our choice of parameters is partly motivated by Sorge and Weller's compendium~\cite{Sorge2019} on graph parameters, 
but we limit our attention to those that are most interesting in the context of \mscAcr{}.
For full definitions of the parameters, 
we refer the reader to Sorge and Weller's manuscript~\cite{Sorge2019} or \cref{sec:paramzoo} in the appendix.

\subsection{Underlying graph parameterization}

\begin{lemma}
	If $\TG=(V,(E_t)_{t=1}^\tau)$ is a temporal graph and every layer $\TG_t=(V,E_t)$ of $\TG$ is bipartite for $t \in \{1,\ldots,\tau\}$, then~$\is_{\undp}(\TG) \geq  2^{-\tau} \abs{V} $.
\end{lemma}
\begin{proof}
	(By induction on $\tau$.)
	If $\tau=1$, then $\TG_U$~is bipartite and the larger color class in any 2\nobreakdash-coloring of~$\TG_U$ forms an independent set containing at least $\frac{1}{2} \abs{V}$ vertices.
	Suppose the claim holds for $\tau-1$.
	Then, the underlying graph of $\TG'=(V,(E_t)_{t=1}^{\tau-1})$ contains an independent set~$X\subseteq V$ of size at least~$2^{-(\tau - 1)} \abs{V}$.
	The graph $(X,\binom{X}{2} \cap E_\tau)$ is bipartite since it is a subgraph of $(V,E_\tau)$.
	Hence, it contains an independent set~$Y$ of size at least $\frac{1}{2}\abs{X} \geq 2^{-\tau} \abs{V}$.
	Then, $Y$ is also an independent set in~$\TG_U$.
\end{proof}

\begin{proposition}
  \label{prop:fpt-is}
	\mscTsc{} is fixed-parameter tractable with respect to $\is_\undp$.
\end{proposition}
\begin{proof}
	If any layer of $\TG$ is not bipartite, then the input can be immediately rejected.
	Otherwise, let~$\TG_U$ be the underlying graph of $\TG$.
	By \cref{obs:bruteforce}, \mscAcr{} can be solved in time~$\bigO^*(2^{\tau \cdot \abs{V}}) \leq \bigO^*(2 ^ {\tau \cdot \is_{\undp}(\TG) \cdot 2 ^\tau})$.
\end{proof}

\begin{proposition}
	\label{prop:pnphard-dom}
	\mscTsc{} is NP-hard even if $\tau=4$, $\dom(\TG_U) \leq 2$, and $\dco(\TG_U)=0$.
	Hence, the problem is para-NP-hard with respect to $\dom_\undp$ and $\dco_\undp$.
\end{proposition}

\begin{proof}
	
	The reduction in \cref{constr:red-from-clique} may be adjusted to prove this claim.
	In the following, we only describe how that construction and the proof of \cref{thm:nphard-smalltau} must be adjusted, rather than restating the entire proof.
	We will use notation defined there. 
	
Let $(G=(V,E),k)$ be the input instance for \prob{Clique} and let $\Delta$ be the maximum degree of $G$. 
	We change the value of $d$ to~$d\coloneqq m - \binom{k}{2} + k \Delta$.
	We also adjust $\ell$ accordingly, so that $\ell = d + 1$ remains true, and $\tau$ is increased to $\tau=4$.
	The layers will be called $E_0,\ldots,E_3$.
	Hence, the instance that the reduction outputs is $(\TG,d)$ with $\TG=(V',E_0,\ldots,E_3)$.
	
We introduce a new vertex set whose sole purpose is to use up budget for changes.
Let $V_4 \coloneqq \{b_1,\ldots,b_{k\Delta}\}$.
	Then, $V' \coloneqq V_1 \cup \ldots \cup V_4$ where $V_1,V_2,V_3$ are defined as in the original reduction.
	The internal edges of $V_4$
	do not change.
	Let 
\begin{align*}
		E^Q \coloneqq 
		E^P \cup \{ \{b_i,b_{i+1}\} \mid i \in \{1,\ldots, k\Delta-1\}\}.
	\end{align*}
	
	The purpose of the initial layer $E_0$ is merely to ensure that the underlying graph has domination number $2$ and is a cograph.
	We achieve this by making the initial layer a complete bipartite graph.
	The other layers are mostly very similar to those defined in the original reduction.
	Let:
	\begin{align*}
	W_1 & \coloneqq \{u_i^v \in V_1, r_i \in V_3, b_i \in V_4 \mid v \in V, i \text{ is odd}\} \\
	W_2 & \coloneqq V' \setminus W_1.
	\end{align*}
	Then, 
	\begin{align*}
		E_0  \coloneqq & \{\{x,y\} \mid x \in W_1,y\in W_2 \} \\
		E_1 \coloneqq & E^Q \cup \{ \{r_1,w^e_2\} \mid e \in E\} \cup \{ \{r_2,u^v_1\} \mid v \in V\} \cup \{ \{r_1,b_2\} \}, \\
		E_2 \coloneqq & E^Q \cup \{ \{r_1,w^e_2\} \mid e \in E\} \cup \{ \{u_1^{v_i},w^e_1\}  \{u_1^{v_j},w^e_2\} \mid e=\{v_i,v_j\}\in E, i<j \} \\ & \quad \cup \{ \{r_1,a_1 \}, \{r_1,b_1 \} \}, \\
		E_3 \coloneqq & E^Q \cup \{ \{ w^e_1,w^e_2 \}, \{ w^e_2,w^e_3 \} \mid e \in E  \} \cup \{\{r_1,b_2 \} \}.
	\end{align*}
	
	We start by showing that $\dom_{\undp}(\TG) \leq 2$.
	This follows from the fact that $\TG_1$ is complete bipartite and neither part of the partition is empty.
	Hence, taking a vertex from each part yields a dominating set of size $2$.
	
	Next we will show that $\TG_U$ is a cograph.
	Since $\TG_0$ is complete bipartite with the parts $W_1$ and~$W_2$,
	it suffices to show that $E_1$, $E_2$, and $E_3$ do not contain an induced $P_4$ containing only vertices in $W_1$ or only vertices in $W_2$.
	The only such edges are those between $w^e_1$, $w^e_2$, and $w^e_3$.
	These edges clearly do not form a $P_4$.

	It is easy to see that $\TG$ can be computed in polynomial time.
	
	The correctness proof for the reduction mostly follows along the same lines as in the original reduction.
	We will explain where it must be adjusted.
	
	Suppose that $G$ contains a clique $X\subseteq V$ of size exactly $k$.
	We will give $f_0,f_1,f_2,f_3\colon V' \rightarrow \{1,2\}$ proving that $(\TG,d)$ is a \yes-instance.
	Let $f_0(x) = 1$ if $x\in W_1$ and $f_0(x) = 2$ if $x\in W_2$.
	In the final three layers, the colors of the vertices in $V_1,V_2,V_3$ do not change compared to the colors in the proof of the original reduction.
	The colors of the vertices in $V_4$ are:
	\begin{align*}
		f_1(b_i) \coloneqq f_0(b_i), \quad
		f_2(b_i) \coloneqq 3- f_1(b_i), \quad
		f_3(b_i) \coloneqq 3 - f_2(b_i).
	\end{align*}
	It is easy to see that these colorings are proper.
	We must argue that at most $d$ vertices change color between any two consecutive layers.
	Between the first two layers,
	only the vertices $w_1^e$ if~$v_i \in X$ and $w_3^e$ if~$v_j \in X$ for any~$e=\{v_i,v_j\}$, $i <j$, change colors.
	Since the vertices in $X$ have at most~$\Delta$ incident edges, it follows that the number that change colors is
	$\abs{X} \cdot \Delta = k\Delta \leq d$.
	Between the layers $E_1$ and $E_2$, the only vertices that change colors are those that change colors in the original reduction and the vertices in $V_4$.
	Hence the total number is at most $m - \binom{k}{2} + \abs{V_4} = m - \binom{k}{2} + k\Delta = d$.
	The same thing applies to the changes between the final two layers.
	
	Now suppose that $f_1,\ldots,f_4$ are proper 2-colorings of the layers of $\TG$ such that at most $d$ vertices change colors between consecutive layers.
	Like in the original reduction, no vertex in $V_3$ may change colors.
	This also implies that the vertices in $V_1$ cannot change colors between the layers $E_0$ and $E_1$.
	This implies that the coloring of the vertices in $V_1,V_2,V_3$ in the layer $E_1$ is as in the proof of the original reduction.
	Between the layers $E_1$ and $E_2$ and between the layers $E_2$ and $E_3$, all the vertices in $V_4$ must change colors.
	Hence, only $m-\binom{k}{2}$ vertices in $V_1,V_2,V_3$ may change colors.
	Hence, the same argument as in the original reduction applies.	
\end{proof}

\begin{proposition}
	\label{prop:xp-tw}
	\mscTsc{} can be solved in $\bigO^*(2^{\tau\cdot \tw_{\undp}(\TG)}\cdot (d+1)^{2\tau})$ time.
	Hence, 
	the problem is in~\XP{} when parameterized by~$\tw_{\undp}$.
\end{proposition}

\begin{proof}
	The proof utilizes a standard dynamic programming approach for problems parameterized by treewidth, extending it to the multistage context.
	Let $(\TG=(V,(E_t)_{t=1}^\tau),d)$ be an instance of \mscTsc{}.
	
	Let $\calT = (\calX,T)$ be a tree decomposition of width $\tw_{\undp}(\TG)$ of the underlying graph $\TG_U$ where~$\calX = \{X_1,\ldots,X_r\}$ are the bags associated with the vertex sets $V(X_1),\ldots,V(X_r) \subseteq V$ and $T$ is a rooted tree with vertex set $\calX$.
	Without loss of generality, we may assume that $\calT$ is a nice tree decomposition (for a definition, see, e.g.,~\cite[Sect.~13.1]{Kloks1994}).
	For any $s \in \{1,\ldots,r\}$, let $V_s \subseteq V$ denote the set of all vertices contained in a bag that is part of the subtree of $T$ rooted at $X_s$.
	A \emph{partial multistage two-coloring} (pmt) $f$ of $V' \subseteq V$ is a sequence of functions~$f=(f_1,\ldots,f_\tau)$ with $f_t\colon V'\rightarrow \{1,2\}$.
	We will call $f$ \emph{proper} if each $f_t$ is a proper two-coloring of $(V',E_t \cap \binom{V'}{2})$.
	Note that the number of pmts of $V'$ is at most $2^{\tau \abs{V'}}$.
	The \emph{cost} of $f$ is the vector
	\begin{align*}
		c(f) \coloneqq (\delta(f_1,f_2),\ldots,\delta(f_{\tau-1},f_{\tau})) \in \Nzero^{\tau-1}
	\end{align*}
	if $f$ is proper and $c(f) = \infty$ if it is not.
	If $\tilde{V} \subseteq V'$ and $\tilde{f}$ is a pmt of $\tilde{V}$, then $f$ is an \emph{extension} of~$\tilde{f}$, if $f(v) = \tilde{f}(v)$ for all $v \in \tilde{V}$.
	
	We compute a table $C$ with an entry $C[X_s,f,\delta_1,\delta_2,\ldots,\delta_{\tau-1}]$ for every bag $X_s \in \calX$,
	pmt $f$ of $v(X_s)$, and $\delta_1,\ldots,\delta_{\tau-1} \in \{0,\ldots,d\}$.
	The value of $C[X_s,f,\delta_1,\delta_2,\ldots,\delta_{\tau-1}]$ is $1$ if there is a proper pmt $\tilde{f}$ of $V_s$ that is an extension of $f$ and has $c(\tilde{f}) \leq (\delta_1,\ldots,\delta_{\tau-1})$ (component-wise).
	Otherwise, the value is~$0$.
	Since $\abs{V(X_s)} \leq \tw(G_U)$ for all $s$, 
	the number of entries of $C$ is at most $r\cdot 2^{\tau \cdot \tw(G_U)}\cdot (d+1)^{\tau}$.
	
	We compute $C$ from the leaves of $T$ up. First, assume that $X_s$ with $V(X_s) = \{v\}$ is a leaf node of $T$.
	Then, for any pmt $f$ of $X_s$ and $\delta_1,\ldots,\delta_{\tau-1}\in \{0,\ldots,d\}$, 
	we set $C[X_s,f,\delta_1,\delta_2,\ldots,\delta_{\tau-1}] = 1$ if and only if~$c(f)\leq (\delta_1,\ldots,\delta_{\tau-1})$.
	Now, suppose that $X_s$ is an insertion node, that $X_{s'}$ is its only child, and that $V(X_s) \setminus V(X_{s'}) = \{v\}$.
	For any proper pmt $f$ of $V(X_s)$, let $f'$ be the pmt of $V(X_{s'})$ obtained by deleting~$v$ from the domain.
	For $t \in \{1,\ldots,\tau-1\}$, let $\delta^v_t = 1$ if $f_t(v) \neq f_{t+1}(v)$ and $\delta^v_t = 0$, otherwise.
	Then, $C[X_s,f,\delta_1,\delta_2,\ldots,\delta_{\tau-1}] = C[X_{s'},f',\delta_1-\delta^v_1,\ldots,\delta_{\tau-1}-\delta^v_{\tau-1}]$.
	If $f$ is not proper, then we simply set $C[X_s,f,\delta_1,\delta_2,\ldots,\delta_{\tau-1}] = 0$.
	Next, suppose that~$X_s$ is a forget node, let $X_{s'}$ again be its only child, and let $V(X_{s'}) \setminus V(X_{s}) = \{v\}$.
	For any pmt $f$ of $V(X_s)$, define two extensions $f_1$ and $f_2$ to $V(X_{s'})$ by assigning $v$ the colors $1$ and $2$, respectively.
	Then,~$C[X_s,f,\delta_1,\delta_2,\ldots,\delta_{\tau-1}] = \min_{i\in \{1,2\}} C[X_{s'},f_i,\delta_1,\ldots,\delta_{\tau-1}]$.
	Finally, suppose that $X_s$ is a join node with children $X_{s'}$ and $X_{s''}$ such that~$V(X_s) = V(X_{s'}) = V(X_{s''})$.
	Then, we set $C[X_s,f,\delta_1,\delta_2,\ldots,\delta_{\tau-1}] = 1$ if there are $\delta'_1,\ldots,\delta'_{\tau-1}$ and $\delta''_1,\ldots,\delta''_{\tau-1}$ such that the following conditions hold:
	\begin{inparaenum}[(i)]
		\item $\delta_t \geq \delta'_t + \delta''_t$ for all $t \in \{1,\ldots,\tau-1\}$,
		\item $C[X_{s'},f,\delta'_1,\ldots,\delta'_{\tau-1}] = 1$, and
		\item $C[X_{s''},f,\delta''_1,\ldots,\delta''_{\tau-1}] = 1$.
	\end{inparaenum}

	The input instance $(\TG,d)$ is a \yes-instance if and only if there is an $f$ with $C[X_a,f,d,\ldots,d] = 1$ where $X_a$ is the root of $\calX$.
	
	\emph{Running time: }
	As we mentioned before, the table $C$ has at most $r\cdot 2^{\tau \cdot \tw(G_U)}\cdot (d+1)^{\tau}$ entries.
	Moreover, $r \in \bigO(n)$.
	Computing an entry requires determining whether a pmt is proper and  in the worst case (join nodes) $(d+1)^\tau$ look-ups.
	This leads to a total running time of $\bigO^*( 2^{\tau \cdot \tw(G_U)}\cdot (d+1)^{2\tau})$.
	
	\emph{Correctness}: By induction on the structure of $\calT$.
\end{proof}

\noindent
Note that the running time of this algorithm also implies that \mscTsc{} is \fpt{} with respect to $\tau + d + \tw_{\undp}$.

\begin{proposition}
	\label{prop:pnph-delta-undp}
	\mscTsc{} is NP-hard even if $\tau=3$ and $\Delta(G) =3$.
	Hence, the problem is para-NP-hard with respect to $\Delta_\undp$.
\end{proposition}
\begin{proof}[Proof sketch]
	The proof is an adjustment of \cref{constr:red-from-clique}, similar to the proof of \cref{prop:pnphard-dom}.
	We will only give a brief sketch.
Let $\Delta$ denote the maximum degree of the graph $G$ in the input instance $(G,k)$ for \prob{Clique}.
	We extend the lengths of the paths representing the vertices in $G$ such that they each contain $2\Delta+1$ vertices.
	Then, the edges in $E_2$ between $u^v_1$ and vertices representing the edges incident to~$v$ are connected to $u^v_3,u^v_5,\ldots,u^v_{2\deg(v)+1}$.
	This requires us to change~$d$ to $d\coloneqq \max\{ k(2\Delta +1),m-\binom{k}{2}\}$.
	Then, we introduce a gadget (like in the proof of \cref{prop:pnphard-dom}) to use up the extraneous budget either between layers $E_1$ and $E_2$ or between $E_2$ and $E_3$.
	In the same way, 
	we replace the path $r_1,\ldots,r_{d+1}$ by a longer path in order to reduce the degree of the vertices $r_i$.
\end{proof}

\begin{proposition}
	\mscTsc{} is fixed-parameter tractable with respect to $\vc_{\undp}$.
\end{proposition}
\begin{proof}
	Note that $\dcc \preceq \vc$ and, hence, by \cref{prop:param-preservation}, $\dcc_{\sump} \preceq \vc_{\sump}$.
	Moreover, $\vc$ is a monotonically increasing parameter.
	Therefore, by \cref{prop:param-relationships}, it follows that $\vc_{\sump} \preceq \vc_{\undp}$.
	As we will show in \cref{thm:fpt-dcc-sum}, 
	\mscAcr{} is \fpt{} when parameterized by~$\dcc_{\sump}$.
\end{proof}

\begin{proposition}
	\label{prop:nphard-dbi-undp}
	\mscTsc{} is NP-hard even if $\tau = 3$ and $\dbi_{\undp} = 2$.
\end{proposition}
\begin{proof}
	The claim follows from \cref{constr:red-from-clique} and the proof of \cref{thm:nphard-smalltau}.
	Note that, if $\TG_U$ is the underlying graph of the temporal graph generated by the reduction, then $\TG_U -\{r_1,r_2\}$ is bipartite.
\end{proof}

Next, we will prove that \mscTsc{} is \W{1}-hard with respect to $\fes_{\undp}$.
In fact, we will prove the following slightly stronger statement:

\begin{proposition}
 \label{prop:whardfesUtau}
 \mscTsc{} is \W{1}-hard when parameterized by~$\tau$,
 even if the feedback edge number~$\fes(\TG_U)$ of the underlying graph is constant.
\end{proposition}

\noindent
We already showed that \mscAcr{} is \XP{} regarding~$\tw_{\undp}$ implying that 
it is~\XP{} and \W{1}-hard when parameterized by~$\tw_{\undp}$, 
$\fvs_{\undp}$, 
and $\fes_{\undp}$, 
since~$\tw \preceq \fvs \preceq \fes$.
The following hardness proof is a little more involved than most of the previous ones.
Our reduction is from the following:
\decprob{\prob{Multicolored Clique} (\prob{MC})}{mc}
{A $k$-colored static graph $G=(V,E)$ with $V = V_1 \uplus \ldots \uplus V_k$.}
{Does $G$ contain a clique $X \subseteq V$ such that $\abs{X \cap V_i}=1$ for all $i\in \set{k}$?}
\noindent
\prob{Multicolored Clique} is \W{1}-hard when parameterized by~$k$~\cite{FellowsHRV09,Pietrzak03}.

\begin{construction}
	\label{constr:mc-clique-red}
	Let $(G=(V,E),k)$ with $V = V_1 \uplus \ldots \uplus V_k$ be an instance of \prob{Multicolored Clique}.
	We may assume that $\abs{V_1} = \ldots = \abs{V_k} = n$
	(if color classes do not have the same size, we can add isolated vertices),
	that all $V_i$ are independent,
	and that $\abs{E} \geq \binom{k}{2}$ (otherwise, this is clearly a \no-instance).
	Let~$V_i = \{v^i_0,\ldots,v^i_{n-1}\}$.
	
	We will now describe an instance $(\TG=(V',(E_t)_{t=1}^\tau),d)$ of \mscTsc{} with $\fes(\TG_U) = 2$
	(see~\cref{fig:clique-red-3} for an illustration).
	\begin{figure}
		\includegraphics[width=\textwidth]{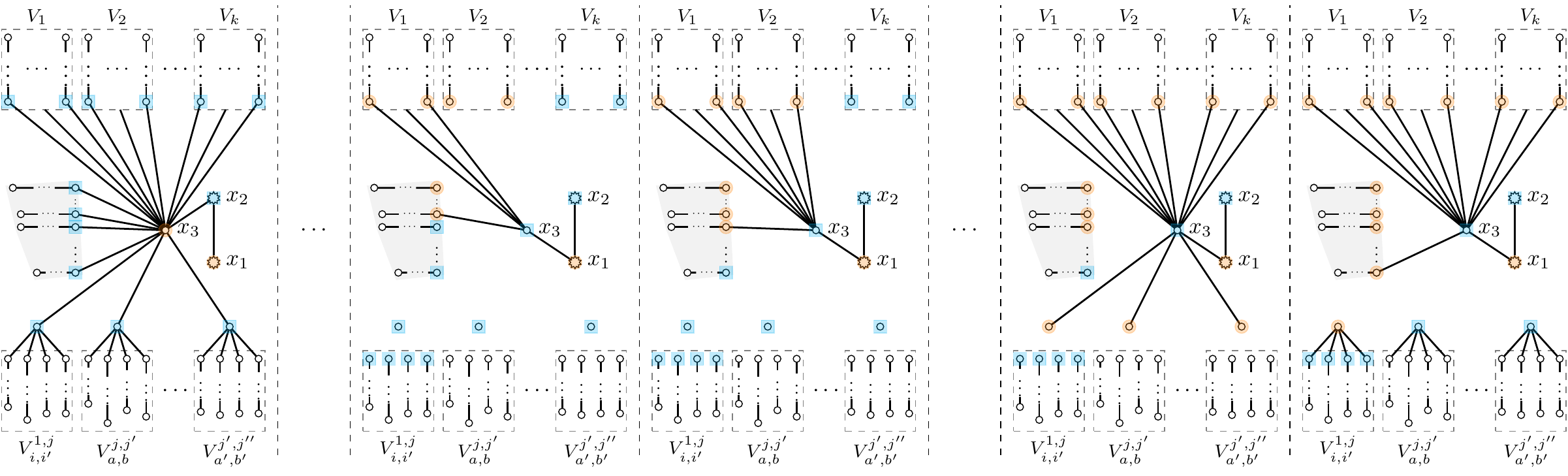}
		\caption{Illustration of \cref{constr:mc-clique-red}.
		Shown are
		the first layer (left),
		the two layers when transitioning from the phase regarding~$V_2$ to the phase regarding~$V_3$ (middle),
		the last two layers (right).
		In the gray area,
		the waste-budget gadget is depicted.
		In this example,
		the edge~$\{v_i^1,v_{i'}^j\}$ is chosen into the clique.
		Note that many vertices (those from paths and stars)
		are not depicted.}
		\label{fig:clique-red-3}
  \end{figure}
	We let $\tau \coloneqq 2 k (k-1) + 3$ and $d \coloneqq \abs{E}$.
	
	The general idea behind the reduction is as follows.
	We consider the steps between consecutive layers and the number of changes to the coloring in those steps.
	The value of $\tau$ implies that there are $2 k (k-1) + 2$ steps in total.
	There are $2k-2$ such steps for each color class in $G$,
	while the final two steps do not correspond to any color class.
	Of the $2k-2$ steps that correspond to $c \in \set{k}$, two will be used to verify adjacency to each of the $k-1$ other color classes.
	In order to be able to refer to these steps easily, we will use the following notation for any $c,c'\in\set{k}$, $c\neq c'$:
	\begin{align*}
		T(c\rightarrow c') & \coloneqq
		\begin{cases}
			2(c-1)(k-1) + c', & \text{ if } c<c',\\
			2(c-1)(k-1) + c' -1, & \text{ if } c > c',
		\end{cases}
		\text{ and }\\
		T(c\Rightarrow c') & \coloneqq T(c\rightarrow c') + k -1
	\end{align*}
	
	We will use several gadgets.
	The first gadget maintains its coloring throughout most of the lifetime of the instance.
	We use it to enforce a particular, predictable coloring on vertices in other gadgets at certain points.
	The second type of gadget represents the selection of a vertex in a certain color class.
	If the vertex $v^i_j$ is to be added to the clique,
	it forces any multistage $2$-coloring to make $j$ changes in the first $k-1$ steps corresponding to the color class $i$ and $n-j-1$ changes in the following $k-1$ steps corresponding to this class.
	There is a third type of gadget.
	Its purpose is to verify that the vertices selected by the first gadget type are pairwise adjacent.
	There are numerous additional vertices whose sole purpose is to ensure that the coloring of vertices cannot change in unexpected ways.
	More specifically, 
	when we say that a vertex $v$ is \emph{blocked} in time step~$t$, 
	we mean that we add $d$ vertices that are adjacent to $v$ in the layers $t-1$ and $t$
	and isolated in all other layers.
	There are also further vertices designed to use up extraneous budget for changes during certain time steps.
	
	We start by describing the first gadget, whose purpose is to maintain a predictable coloring so it can be used to enforce a certain coloring on other parts of the instance at particular points in time.
	This gadget contains the vertices $x_1,x_2,x_3$.
	The edge~$\{x_1,x_2\}$ is present in every layer of $\TG$.
	The edge $\{x_2,x_3\}$ exists only in the first layer, 
	while $\{x_1,x_3\}$ is in the all but the first layers.
	The verices $x_1$ and~$x_2$ are blocked in every step.
	
	Next, we define the second type of gadget, which models the selection of a vertex in a color class.
	The gadget representing a certain color class $V_c$, $c \in \set{k}$, consists of $(n-1)(k-1)$ vertices $w^c_{i,j}$ for $i \in \set{n-1}$, $j \in \set{k-1}$.
	The vertex $w^c_{i,j}$ is blocked in all time steps except for the step $T(c\rightarrow j)$ and the step $T(c\Rightarrow j)$.
	There is an edge between $w^c_{i,j}$ and $w^c_{i,j+1}$ in the layers from $T(c\rightarrow j + 1)$ to $T(c\Rightarrow 1)$ and from $T(c\Rightarrow j)$ to $T(c+1\rightarrow 1)$.
	Additionally, 
  in the very first and in the final layer of $\TG$, 
  all edges $\{w^c_{i,j},w^c_{i,j+1}\}$ are present and there is an edge from $x_3$ to $w^c_{i,1}$ for all $c \in \set{k}$ and $i\in\set{n-1}$.
  Moreover, 
  for every $c \in \set{k}$,
  there is an edge from $x_3$ to $w^c_{i,1}$ for all~$i\in\set{n-1}$ in all layers of index larger than~$T(c\Rightarrow c')$,
  with~$c'=\max\set{k}\setminus\{c\}$.
	This gadget is illustrated in the top part of \cref{fig:clique-red-3-1}.
	
	Next, we will describe the gadget that verifies that vertices selected in the previous gadget are pairwise adjacent.
	There is one such gadget for every edge $e=\{v^c_j,v^{c'}_{j'}\}\in E$, $1\leq c<c'\leq k$, $j,j' \in \set[0]{n-1}$.
	The gadget consists of a root vertex $u^e_0$ and four paths.
	The root is blocked in every step except for the final two.
	There is an edge between $u^e_0$ and $x_3$ in the first and the~$(\tau-2)$nd layer.
	The first vertex of each of the four paths is adjacent to~$u^e_0$ in the first and in the final layer.
	The edges of the paths are present in every layer.
	These paths consist of $n-1-j$, $j$, $n-1-j'$, and $j'$ vertices, respectively.
	The vertices on the path of size $n-1-j$ are blocked in every time step except for step $T(c\rightarrow c')$.
	Those on the path of size $j$ are blocked except for step $T(c\Rightarrow c')$.
	The vertices on the path of size $n-1-j'$ are blocked except for step $T(c'\rightarrow c)$.
	Finally, those on the path of size $j'$ are blocked except for step  $T(c' \Rightarrow c)$.
	
	Finally, there is a gadget whose purpose is to waste extraneous budget for changes.
	It consists of 
	$\tau-2$ paths.
	There are $\tau-4$ paths $P_3,\ldots,P_{\tau-2}$ containing $d-(n-1)$ vertices each,
	one path $P_{2}$ that consists of $d-n$ vertices,
	and one path $P_{\tau}$ that consists of $\binom{k}{2}$ vertices.
	For each~$i\in\set[2]{\tau}\setminus\{\tau-1\}$,
	the first vertex in $P_i$ is adjacent to $x_3$ exactly in the first and $i$th layer,
	where in all but the $i$th layer,
	all vertices from~$P_i$ are blocked.
	\cqed
\end{construction}

\begin{lemma}
	\label{lemma:mc-clique-red}
	The input instance to \cref{constr:mc-clique-red} is a \yes-instance for \prob{Multicolored Clique} if and only if the output instance is a \yes-instance for \mscTsc{}.
\end{lemma}
\begin{proof}
	\RD
	Suppose that $X = \{v^1_{i_1}, \ldots, v^k_{i_k}\}$ with $v^i_{j_i} \in V_i$ is a multicolored clique in $G$.
	We must construct proper 2-colorings $f_1,\ldots,f_\tau \colon V' \rightarrow \{1,2\}$ with $\delta(f_t,f_{t+1}) \leq d$.
	
	We let $f_t(x_1) \coloneqq 1$ and $f_t(x_2) \coloneqq 2$ for all $t \in\set{\tau}$.
	The coloring of $x_3$ is $f_1(x_3) \coloneqq 1$ and $f_t(x_3)\ceq 2$ for all $t\in\set[2]{\tau}$.
	Next, 
	we consider the vertices that are part of the second type of gadget
	(see~\cref{fig:clique-red-3-1} for an illustrative example).
	\begin{figure}
		\includegraphics[width=\textwidth]{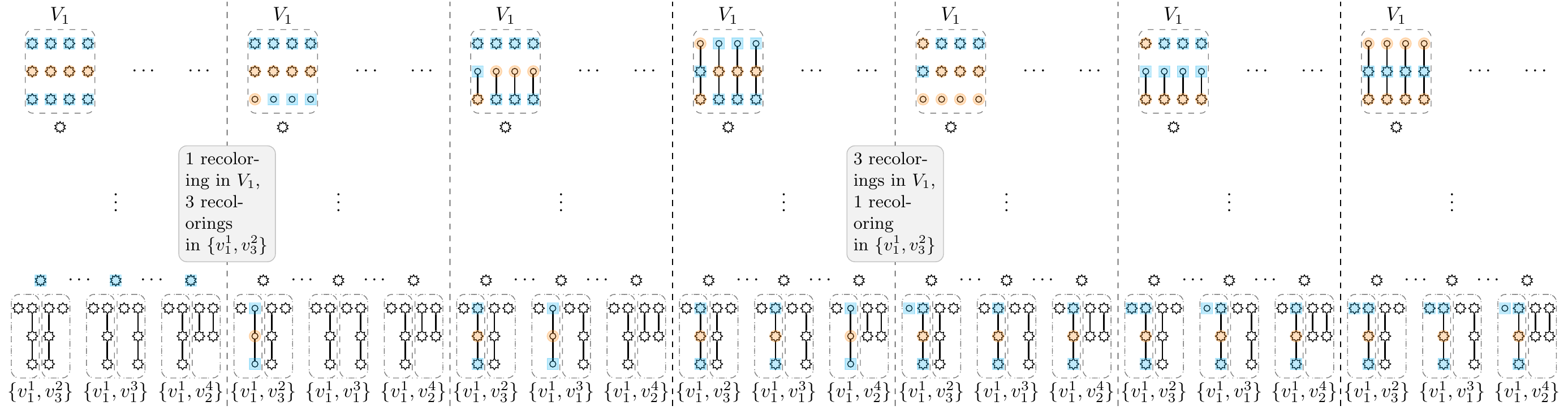}
		\caption{Illustrative example of the recolorings in \cref{constr:mc-clique-red}.
		Here,
		$n=5$ and~$k=4$.
		The recolorings here represents the case that vertex~$v_1^1$ is chosen into the clique,
		together with its incident edges to~$v_3^2$, $v_1^3$, and~$v_2^4$.}
		\label{fig:clique-red-3-1}
	\end{figure}
	Let $c\in\set{k}$, then for any $t \in \set{\tau}$, we let:
	\begin{align*}
		f_t(w^c_{i,j}) \coloneqq
		\begin{cases}
		1 + (j \bmod 2) , & \text{ if } t \leq T(c\rightarrow j),\\
		1 + (j \bmod 2) , & \text{ if }  T(c\rightarrow j)+ 1 \leq t \leq T(c\Rightarrow j) \text{ and } i > i_c, \\
		1 + (j +1 \bmod 2) , & \text{ if }  T(c\rightarrow j)+ 1 \leq t \leq T(c\Rightarrow j) \text{ and } i \leq i_c, \\
		1 + (j+1 \bmod 2), & \text{ if } t > T(c\Rightarrow j) +1.
		\end{cases}
	\end{align*}
	\noindent
	Next, we will give the coloring of the vertices in the gadget that verifies adjacency within the clique.
	First, consider an edge $e$ that has at least one endpoint outside of the clique.
	Then, $f_t(u^e_0) \coloneqq 2$ for all $t \in \set{\tau-2}\cup\{\tau\}$ and $f_{\tau-1}(u^e_0) \coloneqq 1$.
	If $y$ is the the $i$-th vertex on one of the four paths in the gadget, then $f_t(y) \coloneqq 1 + (i + 1 \bmod 2)$ for $t\in\set{\tau}$.
	Next, consider an edge $e=\{v^c_{i_c},v^{c'}_{i_{c'}}\}$ with $c \neq c'$ that has both endpoints in $X$.
	Then, $f_t(u^e_0) \coloneqq 2$ for all $t \in \set{\tau-2}$ and $f_t(u^e_0) \coloneqq 1$ for $t \in \{\tau-1,\tau\}$.
	Let $y^{T(c\rightarrow c')}_i$, $y^{T(c\Rightarrow c')}_i$, $y^{T(c'\rightarrow c)}_i$, and $y^{T(c'\Rightarrow c)}_i$  be the $i$-th vertex on the path containing $n-1-c$, $c$, $n-1-c'$, and $c'$ vertices, respectively.
	Then, the coloring of this vertex is:
	\begin{align*}
		f_t(y^{t'}_i) \coloneqq
		\begin{cases}
			1 + (i + 1 \bmod 2), & \text{ if } t\leq t', \\
			1 + (i \bmod 2), & \text{ if } t> t',
		\end{cases}
	\end{align*}
	for $t' \in \{T(c\rightarrow c'),T(c\Rightarrow c'), T(c'\rightarrow c),T(c'\Rightarrow c)\}$.
	As to the gadget used to waste budget,
	The vertices on the paths $P_i$ only change (completely) their colors in the $i$th layer.
	Finally, any vertex $u$ introduced to block another vertex $v$ in step $t$ receives the coloring $f_{t'} (u) \coloneqq 3 - f_{t}(v)$ for all $t' \in\set{\tau}$.
	
	It is easy to verify that $f_1,\ldots,f_\tau$ are proper 2-colorings.
	We must show that $\delta(f_t,f_{t+1}) \leq d$ for all $t \in \set{\tau-1}$.
	First, consider $t=T(c\rightarrow c')$.
	The number of vertices that change their colors in order to use up budget between $f_t$ and $f_{t+1}$ is $d-(n-1)$.
	The vertices $w^c_{i,c'}$ for all $i \in \set{i_c}$ also change colors.
	Finally, the only vertices in an adjacency verification gadget that change colors are in the gadget for the edge $\{v^c_{i_c}\}$.
	All vertices in the path containing $n-i_c-1$ vertices change their colors.
	No other vertex changes its colors.
	That accounts for a total of $d-(n-1) + i_c + (n-i_c -1) = d$ changes.
	The argument for $t=T(c\Rightarrow c')$ is analogous.
	Next, consider $t = \tau-2$.
	Between $f_t$ and $f_{t+1}$, only~$y^e_0$ for all $e\in E$ change colors,
	that is,
	$d$ changes.
	For $t=\tau-1$, the vertices that change colors between $f_t$ and $f_{\tau+1}$ are $\binom{k}{2}$ vertices that waste budget and all $y^e_0$ for any $e \in E$ that does not have both endpoints in~$X$.
	This accounts for $\binom{k}{2} + (d - \binom{k}{2}) = d$ changes.
	
	\LD{}
	Let $f_1,\ldots,f_\tau \colon V' \rightarrow \{1,2\}$ be proper $2$-colorings of the layers of $\TG$ with $\delta(f_t,f_{t-1}) \leq d$ for all~$t \in \set{\tau-1}$.
	
	We note that if a vertex $v \in V'$ is blocked in time step $t$, then $f_{t-1}(v) = f_{t}(v)$.
	Otherwise, the~$d$~vertices adjacent to $v$ in both $E_{t-1}$ and $E_t$ would also have to be re-colored for a total of~$d+1$~vertices that change colors.
	
	\Wilog{}, we may assume that $f_1(x_1) = 1$.
	This implies that $f_t(x_1) = 1$ and $f_t(x_2) = 2$ for all $t \in \set{\tau}$, since $x_1$ is blocked in all steps.
	This fact, 
	in turn, 
	means that $f_1(x_3) = 1$ and $f_{t}(x_3) =2$ for all $t\in \set[2]{\tau}$.
	Hence, all vertices on the path $P_i$ must change colors between $f_{i-1}$ and $f_i$.
	For $u^e_0$, 
	$e\in E$, 
	we have that $f_t(u^e_0) = 2$ for all $t\in\set{\tau-2}$.
	Therefore, $f_{\tau-1}(u^e_0) = 1$.
Hence, in the remaining gadgets, at most~$n-1$ vertices may change colors between the layers $T(c\rightarrow c')$ and $T(c\rightarrow c')+1$, 
	no vertex may be re-colored between layers $\tau-2$ and $\tau-1$, 
	and at most~$d - \binom{k}{2}$ may change between layers $\tau-1$ and $\tau$.
	
	Because $f_1(x_3) = 1$, it follows that $f_1(w^c_{i,j}) = 1 + (j \bmod 2)$ for all $c\in\set{k}$, $i\in\set{n-1}$, and $j \in\set{k-1}$.
	Because $f_t(x_3) = 2$ for all $t>1$, it follows that $f_{t'}(w^c_{i,j}) = 1 + (j + 1 \bmod 2)$ for all $t' > T(c\Rightarrow c')$ with $c' = \max\set{k}\setminus{c}$.
	Since $w^c_{i,j}$ is blocked in all other steps, it must change colors between the layers $T(c\rightarrow j)-1$ and $T(c\rightarrow j)$ or between the layers $T(c\Rightarrow j) -1$ and $T(c\Rightarrow j)$.
	
	For any $c\in\set{k}$, let $i_{c,j}\coloneqq \abs{\{i \in\set{n-1} \mid f_{T(c\rightarrow j)-1} (w^c_{i,1})\neq f_{T(c\rightarrow j)} (w^c_{i,1}) \}}$ denote the number of vertices in the color class gadget re-colored in step $T(c\rightarrow j)$.
	We claim that $i_{c,j}=i_{c,j+1}$ for all $j\in\set{k-2}$ and, hence, $i_{c,j}=i_{c,j'}$ for all $j,j'\in\set{k-1}$.
	This follows from the fact that, if $w^c_{i,j}$ is re-colored in step $T(c\rightarrow j)$, then its color is $1+(j+1 \bmod 2)$ in layer $T(c\rightarrow j+1)$.
	Moreover, the edge $\{w^c_{i,j},w^c{i,j+1}\}$ appears in layer $T(c\rightarrow j+1)$.
	The color of $w^c_{i,j}$ cannot change in step $T(c\rightarrow j+1)$, since this vertex is blocked.
	Hence, $f_{T(c\rightarrow j+1)}(w^c_{i,j+1}) = 1 + (j \mod 2)$.
	Since $f_{T(c\rightarrow j)}(w^c_{i,j+1}) = 1 + (j + 1 \mod 2)$, it follows that $w^c_{i,j+1}$ changes colors in this step.
	By a similar argument, $w^c_{i,j+1}$ cannot change colors in step $T(c\rightarrow j+1)$ if the color of $w^c_{i,j}$ does not change in step $T(c\rightarrow j)$.
	This implies the claim.
	We also note that this implies that $n-c_{i,j} -1$ vertices must change colors in step $T(c\Rightarrow j)$.
	We let~$i_c \coloneqq i_{c,1}$ and $X \coloneqq \{v^1_{i_1},\ldots,v^k_{i_k}\}\subseteq V$.
	We will show that $X$ is a clique.
	
	Let $F_i \coloneqq \{e \in E \mid f_\tau(u^e_0) = i \}$ for $i \in \{1,2\}$.
	Since, as we mentioned before, $f_{\tau-1}(u^e_0) = 1$ for all $e\in E$ and at most $d - \binom{k}{2}$ of these vertices can change colors between the layers $\tau-1$~and~$\tau$, it follows that $\abs{F_2} \leq \abs{E} - \binom{k}{2}$ and, thus, $\abs{F_1} \geq \binom{k}{2}$.
	We will show that the graph $(X,F_1)$ is complete.
	Since $F_1$ contains $\binom{k}{2}$ edges, it suffices to prove that all edges in $F_1$ have both of their endpoints in $X$.
	
	First, we consider the coloring of the four paths in the gadget representing an edge $e\in E^1$ with $e=\{v^c_i,v^{c'}_{i'}\}$.
	Let $y_i$ denote the $i$-th vertex on one of these paths.
	Since $y_1$ is adjacent to $u^e_0$ in the first layer and $f_1(u^e_0)=2$, it follows that $f_1(y_i)=1 + (i+1 \bmod 2)$.
	However, since $e\in F_1$, we have that $f_\tau(u^e_0) = 1$ and thus $f_\tau(y_i) = 1 + (i \bmod 2)$.
	Hence, the color of each of the four paths must change.
	Now, assume that one of the endpoints of $e$, \wilog{} $v^c_i$, is not in $X$.
	Then, $i \neq i_c$. 
	First, suppose that $i < i_c$.
	Consider the path of length $n-i-1$ that is part of the gadget for $e$.
	As we argued before, the color of every vertex on this path must change in some step.
	Since all of the vertices are blocked in every step but $T(c\rightarrow c')$, it follows that they must change colors in this step.
	In that step, $i_c$ vertices in the gadget for color class $c$ are recolored.
	But then, $i_c + n - i -1 > i_c + n - i_c -1 = n -1$ vertices in those two gadgets are colored in step $T(c\rightarrow c')$.
	This contradicts the fact that at most $n-1$ in these gadgets may change colors in that step.
	A similar argument, but involving the step $T(c\Rightarrow c')$, applies if $i>i_c$.
	This proves that $i=i_c$ and therefore $v^c_i \in X$.
	Hence, both endpoints of all edges in $F_1$ are in $X$ and, therefore, $X$ is a clique with one vertex in each color class.
\end{proof}

This allows us to prove \cref{prop:whardfesUtau}.

\begin{proof}[Proof of \cref{prop:whardfesUtau}]
	It is easy to see that \cref{constr:mc-clique-red} can be computed in polynomial time.
	Moreover, the edge $\{x_1,x_2\}$ forms a feedback edge set of size $1$ in the underlying graph of $\TG$, the temporal graph output by \cref{constr:mc-clique-red}.
	This along with \cref{lemma:mc-clique-red} implies the claim.
\end{proof}

\subsection{Maximum parameterization}

We turn our attention to the parameterized complexity of \mscTsc{} with respect to several structural parameters under the maximum parameterization.
We begin with $\ncc_{\maxp}$, the maximum number of connected components over all layers.
Observe that under any 2-coloring the color of a single vertex determines the coloring of its entire connected component.

\begin{observation}
	Every 2-colorable static graph with~$N$ connected components admits exactly~$2^N$ different 2-colorings.
\end{observation}

\noindent
This implies that \mscAcr{} is fixed-parameter tractable with respect to~$\ncc_{\maxp}$.

\begin{proposition}
	\label{prop:fpt-cc}
	\mscTsc{} admits an~$\O(4^{\ncc_{\maxp}(\TG)}\tau)$-time algorithm.
\end{proposition}
\begin{proof}
	Let $N\coloneqq \ncc_{\maxp}(\TG)$.
	We create an auxiliary static directed graph in the following manner.
	For each layer of $\TG$,
	we include a node for every one of the at most~$2^N$ many 2-colorings of this layer.
	There is a directed edge from a node representing a coloring of $\TG_{t}$ to a node representing a coloring of $\TG_{t+1}$ if the recoloring cost between the two is at most $d$.
	Finally, add two nodes $s,t$ and connect $s$ to every node corresponding to a coloring of the first layer and connect  every node that corresponds to a coloring of the final layer to $t$.
	Then, $(\TG,d)$ is a \yes-instance if and only if the auxiliary graph contains a path from $s$ to $t$.
	Moreover, the auxiliary graph contains at most~$\O(4^{\ncc_{\maxp}(\TG)}\tau)$ edges.
\end{proof}

This result is essentially a stronger version of the statement in \cref{cor:fromms2sat} that \mscTsc{} is fixed-parameter tractable with respect to $n$, the number of vertices.
However, 
$\ncc$ and larger parameters are the only structural parameters that yield fixed-parameter tractability with respect to the maximum parameterization.

\begin{proposition}
	\label{prop:pNPh-max}
	\mscTsc{} is NP-hard even for constant values of $\dcc_{\maxp}$, $\vc_{\maxp}$, $\fes_{\maxp}$, and~$\bw_{\maxp}$.
\end{proposition}

\begin{proof}
	By \cref{thm:nphard-fewedges}, \mscAcr{} is NP-hard even if each layer contains at most three edges and the maximum degree in each layer is at most one.
	For temporal graphs $\TG$ with this property, $\dcc_{\maxp}(\TG), \vc_{\maxp}(\TG) \leq 3$, $\bw_{\maxp}(\TG) \leq 1$, and~$\fes_{\maxp}(\TG) = 0$.
\end{proof}

We note that \cref{prop:no-pk} implies that \mscAcr{} does not admit a polynomial kernel for any parameter $p$ listed in \cref{fig:param-hier-results}, since $n \succeq p_{\maxp}$ for all of these parameters.

\subsection{Sum parameterization}

We start with the parameterized complexity of \mscTsc{} with respect to several structural parameters under the sum parameterization.
For $\ncc_{\sump}$, fixed-parameter tractability follows from that for~$\ncc_{\maxp}$.

We start by proving that \mscAcr{} is \fpt{} with respect to the distance to co-cluster under the sum parameterization.
This stands in contrast to the maximum parameterization (see \cref{prop:pNPh-max}).
A graph is a co-cluster if and only if it does not contain $K_2 + K_1$ as an induced subgraph.
By a general result obtained by Cai~\cite{Cai1996}, 
this implies that the problem of determining whether~$\dcc(G) \leq k$ for a static graph~$G$ is fixed-parameter tractable with respect to~$k$.
We will make use of the following fact:

\begin{observation}
  \label{obs:cocl-edge-con}
	If $G$ is a co-cluster, then $G$ is edgeless or connected.
\end{observation}

\begin{theorem}
	\label{thm:fpt-dcc-sum}
	\mscTsc{} is fixed-parameter tractable with respect to $\dcc_{\sump}$.
\end{theorem}

\noindent
We will use the following as an intermediate problem.

\decprob{\msceTsc{} (\msceAcr)}{msp2ce}
{A temporal graph~$\TG=(V,(E_t)_{t=1}^\tau)$, proper partial 2-colorings~$f_1,\dots,f_\tau\colon V\to \{1,2\}$, and an integer~$d\in\Nzero$.}
{Are there 2-coloring extensions~$\bar{f}_1,\dots,\bar{f}_\tau$, 
where~$\bar{f}_t$ is the extension of~$f_t$ for every~$t\in\set{\tau}$,
such that $\bar{f}_t$ is a proper 2-coloring of~$(V,E_t)$ for every~$t\in\set{\tau}$ 
and $\delta(f_t,f_{t+1})\leq d$ for every~$t\in\set{t-1}$?}

\noindent
We have the following immediate reduction rule.

\begin{rrule}
 \label{rr:coleps}
 If an edge~$e$ has two colored endpoints, 
 then delete~$e$.
\end{rrule}

\begin{lemma}
 \label{lem:coleps-poly}
 \msceTsc{} is polynomial-time solvable if the input does not contain any edges.
\end{lemma}

\begin{proof}
	We reduce \msceTsc{} with no edges to the following job scheduling problem:	
	
	\decprob{$(1\mid r_j, p_{j} = 1 \mid L_{\max})$ Scheduling}{sched}
	{A list of jobs $j_1,\ldots,j_n$, 
	where each job $j_i=(r_i,d_i)$ has a release date $r_i \in \Nzero$ and a due date $d_i\in \Nzero$,
	and a maximum lateness $L \in \Nzero$.}
	{Is there a schedule $s \colon \{j_1,\ldots,j_n\} \rightarrow \Nzero$ such that
		\begin{inparaenum}[(i)]
			\item $s(j_i) \neq s(j_{i'})$ if $i\neq i'$,
			\item $s(j_i) \geq r_i$ for all $i \in \{1,\ldots,n\}$, and
			\item $s(j_i) - d_i \leq L$ for all $i\in \{1,\ldots,n\}$?
		\end{inparaenum}
	}
	\noindent
 	Horn~\cite[Sect.~2]{Horn1974} showed that this variant of the scheduling problem
 	can be solved 
	by a polynomial-time greedy algorithm that always schedules the available job with the earliest due date.
	Let $(\TG=(V,(\emptyset)_{t=1}^\tau),f_1,\ldots,f_\tau,d)$ be an instance for \msceAcr{}.
We will say that  vertex $v\in V$ between $t_1, t_2 \in\{1,\ldots,\tau\}$ is \emph{forced to be re-colored} $i \in\{1,2\}$ if:
	\begin{inparaenum}[(i)]
		\item $t_1 < t_2$ and there is no $t_3$ with $t_1 < t_3 < t_2$ such that $f_{t_3}(v)$ is defined,
		\item $f_{t_2}(v) = i$, and
		\item $f_{t_1}(v) = 3 -i$.
	\end{inparaenum}
	Let $\mathcal{R} \subseteq V \times \{1,\ldots,\tau-1\} \times \{2,\ldots,\tau\} \times \{1,2\}$ be the set of all forced re-colorings.
	Specifically,~$(v,t_1,t_2,i) \in \mathcal{R}$ if and only if $v$ is forced to be re-colored $i$ between $t_1$ and $t_2$.

	In the machine scheduling model, 
	only one job can be performed per time step, 
	but, 
	in a solution for an \mscAcr{} instance, 
	up to $d$ vertices can be re-colored.
	Hence, 
	we each transition between two layers with $d$ time slots.
	For~$t\in\{1,\ldots,\tau-1\}$, the time slots $d(t-1)+1,\ldots, dt$ correspond to changes in the coloring between the layers $t$ and $t+1$.
	For any forced re-coloring~$(v,t_1,t_2,c) \in \calR$, we create a job $j_i$ with release date $r_i = d(t_1-1)+1$ and due date $d_i = dt_2$.
	We will show that the given instance of \msceAcr{} admits a solution if and only if this set of jobs admits a schedule with maximum lateness~$0$.
	
	\RD{}
	Suppose that $\bar{f}_1,\ldots,\bar{f}_\tau$ is a solution to the instance that extends $f_1,\ldots,f_\tau$.
	It is easy to see that, 
	if $(v,t_1,t_2,i) \in \mathcal{R}$, 
	then $f_{t_1}(v) \neq f_{t_2}(v)$.
	Hence, 
	there must be a $t$ with $f_{t}(v) \neq f_{t+1}(v)$ and $t\in\set[t_1]{t_2-1}$.
	Then, a machine schedule for the instance described above can be constructed by scheduling the job corresponding to $(v,t_1,t_2,i)$ in one of the slots $d(t-1)+1,\ldots,dt$.
	Since $\delta(f_t,f_{t+1}) \leq d$, there are enough slots.
	
	\LD{}
	Suppose that we are given a machine schedule with maximum lateness $0$ for the aforementioned instance.
	We construct an initial coloring~$\bar{f}_1$ by assigning each vertex~$v$ the color~$i$, 
	if there is a $t\in\{1,\ldots,\tau\}$ such that $f_t(v)=1$ and $f_{t'}(v)$ is undefined for all $t' < t$.
	If $f_t(v)$ is undefined for all $t\in\set{\tau}$, then we assign $\bar{f}_1(v)$ arbitrarily.
	We construct $f_2,\ldots,f_\tau$ as follows.
	We let $f_{t+1}(v) = 3 - f_t(v)$ if the given schedule assigns a job $j_i$ corresponding to a forced re-coloring~$(v,t_1,t_2,3 - f_t(v)) \in \calR$ to a slot between $d(t-1)+1$ and $dt$.
	Otherwise, we let $f_{t+1}(v) = f_t(v)$.
\end{proof}

\noindent
The idea in the proof of~\cref{thm:fpt-dcc-sum} is as follows.
After computing a distance-to-co-cluster set for each layer,
we check for all possible colorings of these sets,
and then propagate the colorings.
We finally arrive at an instance of \msceAcr{} with no edges,
which is decidable in polynomial time.

\begin{algorithm}[t]
  $T^+,T^-\gets\emptyset$\;
  \ForEach{$t\in \set{\tau}$}{\label{alg:xp:forloopx}
    $X_t \gets$ a minimum set such that $G_t - X_t$ is a co-cluster\;
    \myIf{$G_t -X_t$ is connected}{$T^+\gets T^+\cup\{t\}$ \textbf{else} $T^-\gets T^-\cup\{t\}$}
  }
  \ForEach(\tcp*[f]{$2^{\dcc_{\undp}}$~many}){$f_1\colon X_1\to \{1,2\},\dots,f_\tau\colon X_\tau\to\{1,2\}$}{\label{alg:xp:forloopfone}
    \ForEach{$t\in\set{\tau}$}{
      \myIf{$t\in T^+$}{while $\exists \, \{u,v\} \in E_t$ s.t.\ $f_t(u) = i$ and $f_t(v)$ is undefined, let $f_t(v) \gets 3 -i$}
      \myIf{$t\in T^-$}{$F_t\gets\{f_t^1,f_t^2\}$ with the two possible colorings~$f_t^1,f_t^2$ of~$G_t-X_t$}
    }
    \ForEach(\tcp*[f]{$\leq 2^{\tau}$~many}){$(f_{t_1}',\dots,f_{t_{|T^-|}}')\in\bigtimes_{t\in T^-} F_t$}{\label{alg:xp:forloopg}
      Let~$\tilde{f}_t\gets f_t$ if~$t\in T^+$ and~$\tilde{f}_t\gets f_t\cup f_t'$ if~$t\in T^-$\;
      \If{$\tilde{f}_1,\dots,\tilde{f}_\tau$ are proper partial colorings}{
        \If{$(\TG,\tilde{f}_1,\ldots,\tilde{f}_\tau,d)$ is a \yes-instance for \prob{MS2CE}}{
          \Return{\yes}\tcp*[f]{decidable in polynomial~time (\cref{lem:coleps-poly})}\label{alg:xp:yes}
        }
      }
    }
  }

  \Return{\no}
  \caption{FPT-algorithm on input instance~$\TG=(V,(E_t)_{t=1}^\tau),d \in \Nzero$. 
  }
\label{alg:fpt-dcc-sum}
\end{algorithm}

\begin{proof}[Proof of~\cref{thm:fpt-dcc-sum}]
	Let $\I=(G,d)$ be an instance of \mscTsc{}.
	Let $G=(V,(E_t)_{t=1}^\tau)$
	and
	$G_t \coloneqq (V,E_t)$ be the $t$-th layer of $G$.
	Let $k\coloneqq \sum_{t=1}^\tau \dcc(G_t)$.
The following algorithm is summarized in pseudocode in \cref{alg:fpt-dcc-sum}.
	
	For each $t \in \{1,\ldots,\tau\}$, 
  using Cai's algorithm~\cite{Cai1996},
	we can compute in~$2^{\O(k)}\cdot|G_t|^{\bigO(1)}$ time
	a minimum set $X_t \subseteq V$ 
	such that $G_t - X_t$ is a co-cluster.
	Let~$(T^+,T^-)$ be a partition of~$\set{\tau}$ such 
	that~$t\in T^+$ if and only if~$G_t - X_t$ is connected
	(see~\cref{obs:cocl-edge-con}).
	For~$t\in T^+$,
	let~$V_t\ceq V(G_t - X_t)$,
	and for~$t\in T^-$,
	let~$V_t\ceq \{v\in V(G_t - X_t)\mid \deg_{G_t}(v)>0\}$ be the vertices in~$G_t - X_t$ incident to at least one edge in~$G_t$.
	We then iterate over all the at most~$2^{k}$~possible partial 2-colorings of $(X_1,\dots,X_\tau)$.
For every layer $t\in T^+$ there are only two possible 2-colorings of~$G_t - X_t$.
	We iterate over all the at most~$2^\tau$ possible 2-colorings of these layers.
	For every~$t\in T^-$,
	if there is an uncolored vertex~$v$ with a neighbor~$w$ colored~$i\in\{1,2\}$,
	then color~$v$ with color~$3-i$.
	Note that this colors all vertices in~$V_t$.
	Let $\tilde{f}_1,\ldots,\tilde{f}_\tau$ be the resulting partial coloring.
	The important thing to note is that for every~$t\in\set{\tau}$ and every edge in~$E_t$ both its endpoints are colored by~$\tilde{f}_t$.
	If one of $\tilde{f}_1,\ldots,\tilde{f}_\tau$ is not proper,
	we reject the coloring,
	otherwise we proceed as follows.
	
	Construct the instance~$\tilde{\I}=(\TG,(\tilde{f}_t)_{t=1}^\tau,d)$ of~\msceTsc{}.
	Since every edge has two colored endpoints,
	applying~\cref{rr:coleps} exhaustively results in an instance~$\tilde{\I}'=(\TG',(\tilde{f}_t)_{t=1}^\tau,d)$ of~\msceTsc{} where~$\TG'$ contains no edge.
	Hence,
	due to~\cref{lem:coleps-poly},
	we can solve $\tilde{\I}'$ in polynomial-time.
  Thus, 
  the overall running time is in~$\sum_{t=1}^\tau 2^{\O(k)}\cdot|G_t|^{\O(1)}+2^{k+\tau}|\TG|^{\O(1)}$.
	
	Clearly,
	if~$\tilde{\I}'$ is a \yes-instance in one choice,
	then~$\I$ is a \yes-instance of~\mscAcr{}.
	That the opposite direction is correct too
	is also not hard not see.
	Note that every solution~$f_1,\dots,f_\tau$ induces a proper partial coloring~$\tilde{f}_1,\dots,\tilde{f}_\tau$,
	where~$\tilde{f}_t$ is induced on~$V_t\cup X_t$ for every~$t\in\set{\tau}$,
	that we will eventually check.
	Moreover,
	the resulting input to~\msceAcr{} is clearly a \yes-instance: 
	$f_1,\dots,f_\tau$ is a solution to~$(\TG,(\tilde{f}_t)_{t=1}^\tau,d)$.
\end{proof}

\begin{proposition}
	\label{prop:nphard-sump}
	\mscTsc{} is \NP-hard even for constant values of \begin{inparaenum}[(i)]
		\item $\dco_{\sump}$,
		\item $\fes_{\sump}$, and
		\item $\Delta_{\sump}$.
	\end{inparaenum}
\end{proposition}
\begin{proof}
	\begin{enumerate}[(i)]
		\item First, note that if all connected components of a static graph $G$ are complete bipartite, then $G$ is a co-graph.
		Secondly, adding edges to every layer of a temporal graph to make every connected component in every layer complete bipartite does not change the solution to \mscAcr{}.
		Hence, we can apply this modification to the output of \cref{constr:red-from-clique} in order to generate instances in which every layer is a co-graph and $\tau = 3$.
		\item Every layer in the temporal graph $\TG$ generated by \cref{constr:red-from-clique} is acyclic.
		Hence, $\fes_{\sump}(\TG) = \tau = 3$.
		\item Follows from \cref{prop:pnph-delta-undp,prop:param-relationships} and the fact that $\Delta$ is monotonically increasing.
	\end{enumerate}
\end{proof}

\noindent
Our final result on structural parameters concerns $\bw_{\sump}$, that is, bandwidth with the sum parameterization.
We first briefly note the following:

\begin{observation}
	\label{obs:bw}
	Let~$G$ be an undirected graph.
	If every connected component in~$G$ contains at most $k$ vertices, 
	then $\bw(G) \leq k-1$.
\end{observation}

\noindent
We will use this observation to show that \mscTsc{} is para-\NP-hard when parameterized by~$\bw_{\sump}$.

\begin{proposition}
	\label{prop:nphard-sump-bw}
	\mscTsc{} is \NP-hard even for a constant value of $\bw_{\sump}$.
\end{proposition}
\begin{proof}
	\prob{Edge Bipartization} is NP-complete, even when restricted to graphs with maximum degree $3$~\cite{Yannakakis1978}.
	First, note that, in the first layer of the temporal graph $\TG$ output by \cref{constr:maxcut-red}, connected components consist of a vertex $v_i$ as well $u^e_i$ for each edge $e\in E$ incident to $v_i$.
	If we assume that $G$ has maximum degree three, it follows that each such connected component contains at most four vertices.
	Hence, $\bw(\TG_1) \leq 3$ by \cref{obs:bw}.
	In the second layer, connected components cannot contain more than two vertices and, hence, $\bw(\TG_2) \leq 1$ by \cref{obs:bw}.
	In all, it follows that $\bw_{\sump}(\TG) \leq 4$.
\end{proof}

\section{Global budget}
\label{sec:global}
The problem we have considered so far is the multistage version of \textsc{2-Coloring} with a \emph{local} budget.
The solution may only be changed by a certain amount between any two consecutive stages.
Heeger~et~al.~\cite{Heeger2021} started the parameterized research of multistage graph problems on a \emph{global} budget where there is no restriction on the number of changes between any two consecutive layers, but instead a restriction on the total number of changes made throughout the lifetime of the instance.
All graph problems studied by Heeger~et~al. are NP-hard even for constant values of the global budget parameter.
By contrast, we will show that a global budget version of \mscTsc{} is fixed-parameter tractable with respect to the budget.
Formally, the global budget version of \mscTsc{} is:

\decprob{\mscgbTsc{} (\mscgbAcr{})}{mscgb}
{A temporal graph~$\TG=(V,(E_t)_{t=1}^\tau)$ and an integer~$D\in\Nzero$.}
{Are there~$f_1,\dots,f_\tau\colon V\to \{1,2\}$ such that~$f_t$ is a 2-coloring of~$(V,E_t)$ for every~$t\in\set{\tau}$
	and $\sum_{t=1}^{\tau-1}\delta(f_t,f_{t+1})\leq D$?}

\noindent
We start by pointing out that \mscgbAcr{}, 
like the local budget version, 
is \NP-hard.
This follows from \cref{thm:nphard-smalltau}, 
since there is no distinction between a local and a global budget if $\tau = 2$.

\begin{observation}
	\label{thm:nphard-global}
	\mscgbTsc{} is \NP-hard.
\end{observation}

\noindent
In order to show that \mscgbTsc{} is fixed-parameter tractable, we will prove the existence of a parameter-preserving transformation to the \prob{Almost 2-SAT} problem, which is defined by:

\decprob{\prob{Almost 2-SAT} (\prob{A2SAT})}{a2sat}
{A Boolean formula $\varphi$ in 2-CNF and an integer $k$.}
{Is there a set of at most $k$ clauses whose removal from $\varphi$ makes the formula satisfiable?}

\noindent
Razgon and O'Sullivan~\cite{Razgon2009} prove that \prob{A2SAT} is fixed-parameter tractable when parameterized by~$k$, 
but the fastest presently known algorithm runs in $\bigO ^* (2.3146^k)$ and is due to Lokshtanov~et~al.~\cite{Lokshtanov2014}.
Kratsch and Wahlstr\"{o}m~\cite{Kratsch2020} show that this problem admits a randomized polynomial kernel.

\begin{proposition}
	\label{prop:globalbudget-a2sat}
	\mscgbTsc{} parameterized by~$D$
	admits a parameter-preserving transformation to 
	\prob{Almost 2-SAT} parameterized by~$k$.
\end{proposition}
\begin{proof}	
	Let $(\TG,D)$ with $\TG=(V,(E_t)_{t=1}^\tau)$ be an instance of \mscgbAcr{}.
	Let~$k\coloneqq D$ and define a Boolean formula~$\varphi$ in the following manner.
	We use the variables $x^v_t$ for $v\in V$ and $t\in\{1,\ldots,\tau\}$.
	Intuitively, the variable $x^v_t$ represents that the vertex $v$ is colored with $1$ at time step $t$ if this variable is set to true and colored with $2$ if it is set to false.
	For every edge $\{u,v\} \in E_t$, we add~$D+1$ copies of the clauses $(x^u_t \vee x^v_t)$ and $(\neg x^u_t \vee \neg x^v_t)$ to $\varphi$.
	These clauses express that the edge~$\{u,v\}$ should not be monochromatic.
	Let~$\varphi_t \ceq \bigwedge_{\{u,v\} \in E_t} \bigwedge_{i=1}^{D+1} (x^u_t \vee x^v_t) \land (\neg x^u_t \vee \neg x^v_t)$ for every~$t\in\set{\tau}$.
	Additionally, 
	for every $t \in \{1,\ldots,\tau-1\}$ and $v \in V$ we add the two clauses $(\neg x^v_t \vee x^v_{t+1})$ and $(x^v_t \vee \neg x^v_{t+1})$.
	These clauses express that the color of~$v$ should not change between layers $t$ and $t+1$.
	Let~$\varphi_t^+ \ceq \bigwedge_{v\in V} (\neg x^v_t \vee x^v_{t+1}) \land (x^v_t \vee \neg x^v_{t+1})$ for every~$t\in\set{\tau-1}$.
	Clearly, $\varphi\ceq \varphi_\tau\wedge\bigwedge_{t\in\set{\tau-1}} (\varphi_t\wedge \varphi_t^+)$ can be computed in polynomial time from~$(\TG,d)$.
	We claim that $(\TG,D)$ admits a multistage 2-coloring with a global budget of at most~$D$ 
	if and only if 
	there is a size-at-most-$k$ subset of the clauses of~$\varphi$ whose removal makes the formula satisfiable.
	
	\RD{}
	Suppose that the temporal graph $\TG$ admits a multistage 2-coloring $f_1,\ldots,f_\tau \colon V \rightarrow \{1,2\}$ such that $\sum_{t=1}^\tau\delta(f_t,f_{t+1})\leq D$.
	We start by giving a truth assignment of the variables of $\varphi$.
	Let:
	\begin{align*}
		\alpha (x^v_t) \coloneqq
		\begin{cases}
			\top, & \text{ if } f_t(v) = 1 \\
			\bot, & \text{ if } f_t(v) = 2.
		\end{cases}
	\end{align*}
	Observe that because the colorings~$f_1,\ldots,f_\tau$ are proper,
	$\varphi_t$ is satisfied for every~$t\in\set{\tau}$.
	We continue by giving a set $C$ of clauses that are to be removed from $\varphi$.
	If $f_t(v) = 1$, but~$f_{t+1}(v) =2$, then we add the clause $(\neg x^v_t \vee x^v_{t+1})$ to $C$.
	Conversely, 
	if $f_t(v) = 2$, 
	but $f_{t+1}(v) =1$, 
	then we add the clause $( x^v_t \vee \neg x^v_{t+1})$.
	Note that $\abs{C} = \sum_{t=1}^\tau|\{v\in V\mid f_{t}(v)\neq f_{t+1}(v)\}| \leq D = k$.
	Hence,
	the assignment $\alpha$ also satisfies all clauses in $\varphi_t^+$ that are not in $C$,
	for every~$t\in\set{\tau-1}$.
	It follows that~$\alpha$ satisfies~$\varphi$.
	
	\LD{} 
	Suppose that $C$ is a set of at most $k$ clauses from $\varphi$ and $\alpha$ is a truth assignment that satisfies all clauses in $\varphi$ that are not in $C$.
	We derive a multistage coloring of $G$ by setting:
	\begin{align*}
				f_t(v) \coloneqq
		\begin{cases}
			1, & \text{ if } \alpha (x^v_t) = \top \\
			2, & \text{ if } \alpha (x^v_t) = \bot.
		\end{cases}
	\end{align*}
	First, note that $C$ cannot contain all clauses representing an edge, since $\abs{C} \leq k=D$.
	Hence, for every $t \in \{1,\ldots,\tau\}$, 
	since~$\varphi_t$ is satisfied,
	$f_t$ is a proper coloring of $(V,E_t)$.
	We must show that there are at most~$D$ changes in the coloring.
	Suppose that the vertex~$v$ changes colors between layers~$t$ and $t+1$.
	If~$f_t(v) = 1$ and $f_{t+1}(v) = 2$, then $\alpha$ does not satisfy the clause $(\neg x_t^v \vee x_{t+1}^v)$, so this clause must be in $C$.
	Similarly, if $f_t(v) = 2$ and $f_{t+1}(v) = 1$, then the clause $( x_t^v \vee  \neg x_{t+1}^v)$ must be in $C$.
	Since~$\abs{C} \leq k$, this implies that there can be at most $k=D$ such color changes.
\end{proof}

\noindent
This directly implies the following:

\begin{corollary}
	\label{cor:fpt-global}
	\mscgbTsc{} parameterized by~$D$
	is fixed-parameter tractable and admits a randomized polynomial kernel.
\end{corollary}

\noindent
We briefly note that the approach described here for \mscAcr{} can also be used to reduce a global budget version of the more general \prob{Multistage 2-SAT} to \prob{Almost 2-SAT}, 
proving the following:

\begin{observation}
	\prob{Multistage 2-SAT on a Global Budget} parameterized by the total number of allowed changes is fixed-parameter tractable and admits a randomized polynomial kernel.
\end{observation}

{\small
\bibliography{strings-long,ms2col}
}

\section{Parameter zoo}
\label{sec:paramzoo}
If $\calC$ is a class of static graphs and $G=(V,E)$ a static graph, then $X\subseteq V$ is a \emph{$\calC$-modulator} in $G$ if $G-X \in \calC$.

Let $G=(V,E)$ be a static graph.

\subparagraph{Bandwidth ($\bw$): }
Let $S_n$ denote the set of all permutations of $\{1,\ldots,n\}$ and assume that $V=\{v_1,\ldots,v_n\}$.
The \emph{bandwidth} of $G$ is $\bw(G) \coloneqq \min_{\pi \in S_{n}} \max_{\{v_i,v_j\} \in E} \abs{\pi(i) - \pi(j)}$.

\subparagraph{Cliquewidth ($\clw$): }
Let $k \in \N$.
A \emph{$k$-expression}, which evaluates to a graph with vertex labels in~$\{1,\ldots,k\}$, is defined inductively by: \begin{inparaenum}[(i)]
	\item if $i \in \{1,\ldots,k\}$, then $\ell(i)$ is a $k$-expression which evaluates to the graph with a single vertex that is labeled $i$,
	\item if $x_1$ and $x_2$ are $k$-expressions, then $x_1 \oplus x_2$ is a $k$-expression which evaluates to the disjoint union of the evaluations of $x_1$ and~$x_2$,
	\item if $x$ is a $k$-expression and $i,j \in \{1,\ldots,k\}$, $i\neq j$, then $\eta_{i,j}(x)$ is a $k$-expression which evaluates to the graph obtained by adding an edge between every pair of vertices $\{u,v\}$ such that $u$ is labeled $i$ and $v$ is labeled $j$, and
	\item if $x$ is a $k$-expression and $i,j \in \{1,\ldots,k\}$, then $\rho_{i\rightarrow j}$ is a $k$-expression which evaluates to the graph obtained from the evaluation of $x$ by changing all the labels of vertices labeled $i$ to $j$.
\end{inparaenum}
The \emph{cliquewidth} $\clw(G)$ of $G$ is the minimum integer $k$ such that there is a $k$-expression that evaluates to $G$.

\subparagraph{Degeneracy ($\dgn$): }
Let $\delta(H)$ denote the minimum degree of a graph $H$.
The \emph{degeneracy} of $G$ is~$\dgn(G) = \max_{V' \subseteq V} \delta(G[V'])$.

\subparagraph{Distance to bipartite ($\dbi$): }
The parameter $\dbi(G)$ is the size of a minimum $\calC$-modulator if $\calC$ is the set of all bipartite graphs.

\subparagraph{Distance to clique ($\dcl$): }
A graph $H=(V',E')$ is \emph{complete} if $E' = \binom{V'}{2}$.
The parameter $\dcl(G)$ is the size of a minimum $\calC$-modulator if $\calC$ is the set of all complete graphs.

\subparagraph{Distance to co-cluster ($\dcc$): }
A graph $H=(V',E')$ is a \emph{co-cluster} if $V' = V_1\cup\ldots V_k$ for some~$k\in \N$ and $E'=\{\{u,v\} \mid u\in V_i, v\in V_j, i\neq j\}$
The parameter $\dcc(G)$ is the size of a minimum $\calC$-modulator if $\calC$ is the set of all co-clusters.

\subparagraph{Distance to co-graph ($\dco$): }
A graph is a \emph{co-graph} if it does not contain an induced $P_4$.
The parameter $\dco(G)$ is the size of a minimum $\calC$-modulator if $\calC$ is the set of all co-graphs.

\subparagraph{Domination number ($\dom$):}
A vertex set $X\subseteq V$ dominates $G$ if every vertex in $V\setminus X$ has a neighbor in $X$.
The parameter
$\dom(G)$ is the size of a minimum dominating set in $G$.

\subparagraph{Feedback edge number ($\fes$): }
A set of edges $X \subseteq E$ is a feedback edge set if $G-X$ is acyclic.
The parameter $\fes(G)$ is the size of a minimum feedback edge set.

\subparagraph{Feedback vertex number ($\fvs$): }
The parameter $\fvs(G)$ is the size of a minimum $\calC$-modulator if $\calC$ is the set of all acyclic graphs.

\subparagraph{Independence number ($\is$):}
A vertex set $X\subseteq V$ is \emph{independent} if $G[X]$ is edgeless.
The parameter
$\is(G)$ is the size of a maximum independent set in $G$.

\subparagraph{Maximum degree ($\Delta$): }
The parameter $\Delta(G)$ is the maximum degree of $G$.

\subparagraph{Maximum diameter of a connected component ($\cdi$): }
The vertex set $X\subseteq V$ is a \emph{connected component} if~$G[X]$ is connected and there is no edge $\{u,v\} \in E$ with $u\in X$ and $v\in V\setminus X$.
The \emph{distance} between two vertices is the length of a shortest path between them.
The \emph{diameter} of a connected graph is the maximum distance between any two vertices.
The parameter $\cdi(G)$ is the maximum diameter of a connected component in $G$.

\subparagraph{Number of connected components ($\ncc$):}
The parameter $\ncc(G)$ is the number of connected components in $G$.

\subparagraph{Treewidth ($\tw$): }
A \emph{tree decomposition} of $G$ is a pair $(\calX,\calT)$ where $\calX \subseteq 2^V$ and $\calT$ is a tree with node set $\calX$ such that
\begin{inparaenum}[(i)]
	\item $\bigcup_{X \in \calX} X = V$,
	\item for $\{u,v\} \in E$ there is an $X \in \calX$ such that~$u,v \in X$, and
	\item for every $v \in V$ the node set $\{X \in \calX \mid v \in X \}$ induces a subtree of $\calT$.
\end{inparaenum}
The \emph{width} of~$(\calX,\calT)$ is $\max_{X \in \calX} \abs{X} -1$.
The \emph{treewidth} $\tw(G)$ of $G$ is the minimum width of a tree decomposition of $G$.

\subparagraph{Vertex cover number ($\vc$): }
The parameter $\vc(G)$ is the size of a minimum $\calC$-modulator if $\calC$ is the set of all edgeless graphs.

\end{document}